\documentclass{preamble}
\usepackage{xcolor}
\definecolor{tabblue}{RGB}{31,119,180}
\definecolor{taborange}{RGB}{255,127,14}

\allowdisplaybreaks

\definecolor{LightCerulean}{RGB}{118,210,251}
\definecolor{Raspberry}{RGB}{247, 67, 140}
\definecolor{PaleGold}{RGB}{252, 239, 164}
\definecolor{PaleRaspberry}{RGB}{250, 179, 209}
\definecolor{tabblue}{RGB}{ 31, 119, 180 }
\definecolor{taborange}{RGB}{ 255, 127, 14 }
\definecolor{tabgreen}{RGB}{ 44, 160, 44 }
\definecolor{tabred}{RGB}{ 214, 39, 40 }
\definecolor{tabpurple}{RGB}{ 148, 103, 189 }
\definecolor{tabbrown}{RGB}{ 140, 86, 75 }
\definecolor{tabpink}{RGB}{ 227, 119, 194 }
\definecolor{tabgray}{RGB}{ 127, 127, 127 }
\definecolor{tabolive}{RGB}{ 188, 189, 34 }
\definecolor{tabcyan}{RGB}{ 23, 190, 207 }

\begin{document}
\title{From molecular model to tensor model of nematic liquid crystals through entropy decomposition
\footnote{This work is funded by the National Natural Science Foundation of China (Nos.~12225102, T2321001, 12288101).}}
\author{Baoming Shi \footnote{Department of Applied Physics and Applied Mathematics, Columbia University, New York, NY 10027, USA.}
\and Dawei Wu \footnote{School of Mathematical Sciences, Peking University, Beijing 100871, China.}
\and Lei Zhang \footnote{School of Mathematical Sciences, Beijing International Center for Mathematical Research, Center for Quantitative Biology, Center for Machine Learning Research, Peking University, Beijing 100871, China ({zhangl@math.pku.edu.cn})}
\and Pingwen Zhang \footnote{School of Mathematics and Statistics, Wuhan University, Wuhan 430072, China; School of Mathematical Sciences, Peking University, Beijing 100871, China ({pzhang@pku.edu.cn})}}
\maketitle

\begin{abstract}
In the mathematical modeling of nematic liquid crystals, a practical and physically reliable $\vb Q$-tensor model can be derived from Onsager's molecular model with the Bingham closure. However, this procedure leads to a singular entropy term that implicitly depends on $\vb Q$, creating both computational and theoretical difficulties. 
In this paper, we characterize this entropy contribution by splitting it into a singular but explicit leading term and an implicit but regular correction term, the latter of which is proven to be sufficiently regular to be accurately approximated numerically, for example, by neural networks. 
This yields a computationally convenient free energy that can be used for the computation of nematic liquid crystals. Our numerical experiments demonstrate that the resulting free energy can capture the isotropic-nematic phase transition as well as the free-boundary droplet configurations.
\end{abstract}

\paragraph*{Math subject classification.} 76A15, 41A60, 68T07

\section{Introduction}


The \textit{nematic liquid crystal} (NLC) is an intermediate state of matter between solids and liquids, in which positional order is lost, but orientational order is still present \cite{de_gennes_physics_1993}. The molecules of NLCs tend to align in certain locally preferred directions called \textit{directors}, which endow them with anisotropic physical properties that can be applied in opto-electric devices, photonics, and elastomers \cite{urbanski_liquid_2017,wu_unlocking_2026,yin_construction_2020}.

There are several mathematical theories for NLCs, ranging from microscopic molecular models to macroscopic tensor models \cite{wang_modelling_2021}.
Onsager's molecular model \cite{onsager_effects_1949} is the most elementary approach. It describes the director $\vb m\in\BbbS^{d-1}$ in $d$ dimensions ($d=2,3$, which we will abbreviate as 2D and 3D) with its distribution function $f(\vb m)$. The (nondimensional) thermodynamic free energy of a homogeneous solution of NLC is given by
\begin{equation} \label{ons-fe-hom}
\mcalF_0[f]=S[f] + \frac12 \<U\>_f = \int_{\BbbS^{d-1}} \qty[ f\ln(\omega_d f) +  \frac12 f U]\d\vb m,
\end{equation}
where $S[f]=\int_{\BbbS^{d-1}} f\ln(\omega_d f)\d\vb m$ is the entropy relative to the uniform distribution on $\BbbS^{d-1}$ (with $\omega_2=2\pi,\omega_3=4\pi$), $U=U[f]$ is a nonlocal interaction potential, $\<\cdot\>_f$ refers to taking an average with respect to $f$, and $\d\vb m$ is the area measure on $\BbbS^{d-1}$. The molecular model incorporates detailed information about molecular shapes, sizes, and interactions and is therefore regarded as the fundamental model for NLCs, but it suffers from the curse of dimensionality. For example, when $d=3$ and the distribution $f$ varies in the space variable $x$ as well, the molecular model requires a distribution function $f(x,\vb m)$ defined on the five-dimensional domain $\BbbR^3\times\BbbS^2$.
To mitigate the dimensionality, the macroscopic tensor model uses a lower-dimensional order parameter $\vb Q$ to describe the anisotropy, which is defined by
\begin{equation}  \label{Q-def}
    \vb Q=\<\vb{mm}-\frac1d \vb I\>_f\triangleq \int_{\BbbS^{d-1}} \qty(\vb{mm}-\frac1d\vb I)f\d\vb m,
\end{equation}
where $\vb{mm}\triangleq\vb{m}\otimes \vb{m}=\vb{mm}^T$. If $\vb Q$ is computed by \eqref{Q-def} from a distribution function $f$, then it is a symmetric traceless tensor in the physical set
\begin{equation} \label{Q-phy}
    \mathcal{Q}_{p}=\left\{\vb A\in\mcalS_0: \lambda_i(\vb A)\in\left(-\frac1d,\frac{d-1}{d} \right)\right\},
\end{equation}
where $\mcalS_0$ is the space of $d\times d$ symmetric traceless tensors and $\lambda_i(\vb A)$ is the $i$-th eigenvalue of $\vb A$.
The most widely used $\vb Q$-tensor model is the phenomenological Landau-de~Gennes (LdG) model \cite{de_gennes_short_1971,de_gennes_physics_1993}, whose free energy has a convenient polynomial form that facilitates quantitative analysis. However, LdG parameters have no precise physical relations to molecular-level quantities, and the polynomial energy imposes no restriction that $\vb Q\in\mcalQ_p$, so its stationary points may fail to correspond to any underlying microscopic distribution function by \eqref{Q-def}, even near the critical temperature \cite{majumdar_equilibrium_2010}.

To resolve the high dimensionality of the molecular model and the lack of microscopic information in the LdG $\vb Q$-tensor model, one can derive a tensor model from the molecular model based on closure approximations, also known as the \textit{mean-field $\vb Q$-tensor model} \cite{ball_nematic_2010,priestley_introduction_1975}.
The key idea is to map a $\vb Q$-tensor to a special director distribution $f_{\vb Q}$ (called the \textit{closure} of $\vb Q$), such that the consistency relation \eqref{Q-def} holds. Substituting $f_{\vb Q}$ into \eqref{ons-fe-hom} yields a free energy determined by $f_{\vb Q}$, and it can therefore be expressed purely in terms of $\vb Q$ (i.e.~\textit{closed}). The $\vb Q$-tensor formulation of the free energy reduces the dimensionality from the distribution function $f(\vb m)$ by $d-1$, and the consistency relation \eqref{Q-def} ensures that $\vb Q$ is a physically admissible tensor lying in the range $\mcalQ_p$, a feature that the LdG model does not have.

The most common closure model for nematic liquid crystals is the \textit{Bingham closure}, where the closure distribution $f_{\vb Q}$ takes the form of a Bingham distribution \cite{bingham_antipodally_1974} that minimizes the entropy $S[f]$ subject to the constraint \eqref{Q-def}.
The existence, uniqueness and well-posedness of the Bingham closure was given in \cite{ball_nematic_2010,katriel_free_1986}.
In \cite{fatkullin_critical_2005,liu_axial_2005}, the authors proved that if the interaction potential takes the Maier-Saupe form \cite{maier_einfache_1959}, then any critical distribution of the molecular energy \eqref{ons-fe-hom} is exactly a Bingham distribution. Hence, the mean-field $\vb Q$-tensor model with Bingham closure works as a quasi-equilibrium approximation of the molecular model.
In numerical simulations of NLC dynamics \cite{chaubal_closure_1998,feng_closure_1998}, the Bingham closure also stood out among various $\vb Q$-tensor closure models by predicting the exact equilibrium order parameter and more accurately predicting the tumbling behaviour of NLC directors in a flow field.

Despite its advantages, the mean-field $\vb Q$-tensor model with Bingham closure involves great challenges, especially in the computation of the entropy term $S[f_{\vb Q}]$.
First, the closure mapping $\vb Q\mapsto f_{\vb Q}$ is an inverse problem, whose forward problem $f_{\vb Q}\mapsto \vb Q$ is a non-elementary integral. This results in a complicated functional with implicit dependence on $\vb Q$, especially in 3D.
The entropy function also has a singularity when $\vb Q$ is near the boundary of $\mcalQ_p$ and $f_{\vb Q}$ degenerates \cite{ball_nematic_2010}, i.e.~one or more of the eigenvalues of $\vb Q$ approaches the bound $-\frac13$, and $f_{\vb Q}$ is concentrated around a lower-dimensional submanifold of $\BbbS^{d-1}$. In these cases, the entropy blows up, and stability is difficult to maintain. 
For these reasons, the mean-field $\vb Q$-tensor model has seen limited use in practice despite its favourable properties.

Several indirect approaches to the Bingham closure have been proposed.
In \cite{chen_maximum_2021,kent_asymptotic_1987,luo_fast_2018}, the authors designed efficient and accurate algorithms for the forward problem $f\mapsto\vb Q$ in 3D. The inverse problem could potentially be solved via nonlinear solvers such as Newton-Raphson iteration applied to these algorithms, but the cost of accurate computation could be high for degenerate $f$.
In \cite{weady_fast_2022,weady_thermodynamically_2022}, instead of computing the singular mapping $f_{\vb Q}$, the authors used interpolation methods to approximate other regular quantities of interest, such as the third- and fourth-order moments of the Bingham closure $f_{\vb Q}$, and reached a high accuracy. However, they did not address the singularity issue of the entropy, and only computed the forward problem through a numerical quadrature, which could lose precision for degenerate $f$.
In \cite{shi_neural_2026}, the authors used a neural network to approximate the molecular entropy from \eqref{ons-fe-hom} under Bingham closure in 3D, and obtained phenomenologically correct results in numerical tests, but the amount of training data used to characterize the blow-up behaviour of the closure was huge.

To resolve the above issues, we propose the following theorem as the main result of this paper.
\begin{theorem} \label{thm:SQ-asymp}
Let $\vb Q$ be the $d$-dimensional ($d=2,3$) tensor order parameter defined in \eqref{Q-def}, $f_{\vb Q}$ its Bingham closure and $S(\vb Q)=S[f_{\vb Q}]$ the molecular entropy from \eqref{ons-fe-hom}.
Then, we have that
\[S(\vb Q) = -\frac12\ln \det(\vb Q+\frac{\vb I}{d}) + \Delta S(\vb Q),\]
where the correction term $\Delta S(\vb Q)$ is a Lipschitz continuous function on $\mcalQ_p$ with a uniformly bounded gradient.
\end{theorem}

Theorem \ref{thm:SQ-asymp}, which will be proven in Sections \ref{sec:bing2d} and \ref{sec:bing3d}, decomposes the troublesome entropy term under Bingham closure (which depends implicitly on $\vb Q$ and blows up) into two numerically tractable terms --- an explicit (log-determinant) leading term, which contains all the singularity, and an implicit correction term with uniform Lipschitz regularity.
The log-determinant form of the leading term is reminiscent of the ``quasi-entropy'' proposed by \cite{xu_quasi-entropy_2022}. The extraction of this term from the full energy gives a stronger and more precise estimate of the blow-up behaviour of $S(\vb Q)$ than similar works on Bingham closure such as \cite{ball_nematic_2010}. 
With Theorem \ref{thm:SQ-asymp} established, the evaluation of the implicit function $S(\vb Q)$ reduces to computing the correction term $\Delta S(\vb Q)$, which is much easier to handle numerically because of the uniform Lipschitz continuity. Therefore, efficient computation of the mean-field $\vb Q$-tensor model can be achieved.
In this paper, we will approximate the correction term $\Delta S(\vb Q)$ with a neural network in Section \ref{sec:nn}, and predict the isotropic–nematic phase transition as well as free-boundary droplet configurations \cite{wu_analysing_2026,wu_diffuse-interface_2025} in Section \ref{sec:num} to demonstrate the feasibility of our approach.


\section{Entropy decomposition in dimension 2} \label{sec:bing2d}

We begin with the simple 2D case ($d=2$) as a concise demonstration of our basic methods: we first discuss the properties of the Bingham distribution, and then study the entropy decomposition of Theorem \ref{thm:SQ-asymp}.
The 2D setting is of significant interest in the research on liquid crystals as a thin-film limit of 3D systems (where the thickness tends to zero) \cite{canevari_well_2020,golovaty_dimension_2015} and is directly relevant to planar liquid-crystal devices \cite{kitson_controllable_2002,spencer_zenithal_2010}.

\subsection{Bingham distribution in 2D}

For any 2D symmetric traceless tensor $\vb Q\in\mcalQ_p$, the Bingham closure seeks for a \textit{Bingham distribution} \cite{bingham_antipodally_1974} of the form
\begin{equation} \label{bingham-2d}
    f(\vb m)=\frac{1}{2\pi Z} \e^{\vb B:\vb{mm}},\ 
    Z=\frac{1}{2\pi} \int_{\BbbS^1} \e^{\vb B:\vb{mm}}\d\vb m,
\end{equation}
where $\vb B\in\mcalS_0$ ($2\times 2$ symmetric traceless tensor), such that the constraint \eqref{Q-def} holds.
By analysing the Lagrange multiplier problem, it is straightforward that the Bingham closure minimizes the molecular entropy $S[f]$ from \eqref{ons-fe-hom} subject to the constraint of the second-order moment \eqref{Q-def} \cite{ball_nematic_2010}.

One observes that $Z, f$ and thereby $\vb Q$, are all determined by $\vb B$ explicitly. To study the inverse problem of Bingham closure (computing $\vb B$ from $\vb Q$). It is necessary to begin with the forward problem centred around the function $Z(\vb B)$.
We present the following well-known properties (see e.g.~\cite{ball_nematic_2010,li_local_2015}).
\begin{lemma} \label{lem:Z2d}
\begin{enumerate}[\rm(i)]
\item $\ln Z(\vb B)$ is smooth and strictly convex.
\item $\vb Q(\vb B)=\nabla_0 [\ln Z(\vb B)] \in \mathcal{Q}_{p}$ is the gradient in the space $\mcalS_0$, where
\[\nabla_0 f(\vb B) = \nabla f(\vb B) - \frac{\tr(\nabla f(\vb B))}{d}\vb I\]
is the projected gradient onto $\mcalS_0$.
\item The Hessian operator $H$ of $\ln Z$, defined by
\[H\vb E:\vb E = \frac{\pt^2(\ln Z(\vb B+t\vb E))}{\pt t^2},\ \vb E\in\mcalS_0,\]
is strictly positive definite.
\end{enumerate}
\end{lemma}



Lemma \ref{lem:Z2d} shows that the relationship between $\vb B$ and $\vb Q$ is completely characterized by $Z(\vb B)$ and its derivatives, so studying the Bingham distribution is equivalent to studying the function $Z$. 

For all rotations $\vb R\in\mathrm{SO}(d)$, with a change of variable $\vb m\to \vb{R}^T\vb m$, we get that
\begin{equation} \label{rot-inv}
    Z(\vb{R}^T\vb{BR})=Z(\vb B),
\end{equation}
so $Z(\vb B)$ is rotationally invariant. By a well-known result in matrix analysis \cite[Thm.~3.1]{lewis_convex_1996}, a rotationally invariant convex function $f(\vb B)$ of symmetric matrices is equivalent to a symmetric convex function $\bar f$ with respect to the eigenvalues $\lambda(\vb B)=[\mu_1,\mu_2]^T$.
Moreover, the matrix gradient of $\nabla f$ rotates accordingly to $\vb B$:
\begin{equation} \label{rot-inv-grad}
    \nabla f(\vb{R}\diag(\mu_1,\mu_2)\vb{R}^T) =\vb{R} \diag(\nabla_0\bar f(\mu_1,\mu_2))\vb{R}^T,\ \vb R\in\mathrm{SO}(d),
\end{equation}
Taking the projection onto $\mcalS_0$, we find that same relation applies to $\ln Z$ and the projected gradient $\vb Q=\nabla_0(\ln Z)$. Hence, $\vb Q$ has the same eigenframe as $\vb B$, and the mapping $\vb Q(\vb B)$ is essentially determined by the relations between $\lambda(\vb Q)$ and $\lambda(\vb B)$.

Due to frame invariance as discussed above, we can assume w.l.o.g.~that 
\begin{equation} \label{B-Q-diag2}
    \vb B=\mqty[\mu&\\&-\mu],\ \vb Q=\mqty[q&\\&-q]
\end{equation}
are diagonalized, and can be parametrized by simple scalars $\mu,q\in\BbbR$. The physical constraint \eqref{Q-phy} then manifests as $-\frac12 < q < \frac12$.
With a harmless abuse of notation, we rewrite $Z(\vb B)$ as $Z(\mu)$. Then, by \eqref{bingham-2d} we have that
\begin{equation}
    Z(\mu) = \frac{1}{2\pi}\int_0^{2\pi}\e^{\mu (\cos^2\theta-\sin^2\theta)}\d\theta = \mathrm{I}_0(\mu),
\end{equation}
where $\mathrm{I}_0$ is the modified Bessel function of the first kind with integer order \cite{abramowitz_handbook_2013}:
\begin{equation} \label{bessel-I0}
    \mathrm{I}_n(\mu)= \frac{1}{2\pi}\int_0^{2\pi} \e^{\mu \cos\theta} \cos (n\theta) \d\theta.
\end{equation}
The equality $Z(\mu)=\mathrm{I}_0(\mu)$ follows from a simple change of variable.
By Lemma \ref{lem:Z2d}, $\vb Q=\nabla_0(\ln Z(\vb B))$. Taking the variation with respect to $\mu$, we get
\[\delta(\ln Z)=\vb Q:\delta\vb B = 2q\delta\mu.\]
Therefore, $2q=\frac{\d(\ln Z)}{\d\mu}$, and the relation between $q$ and $\mu$ are given by
\begin{equation}\label{q2d-eq-I1}
    q(\mu)=\frac12(\ln \mathrm{I}_0)'(\mu) =\frac{\mathrm{I}_1(\mu)}{2 \mathrm{I}_0(\mu)}.
\end{equation}

We can study the properties of $Z(\mu)$ and $q(\mu)$ with the help of Bessel functions. 
\begin{lemma}[{\cite{abramowitz_handbook_2013,fatkullin_critical_2005}}] \label{lem:bessel}
The Bessel functions $\mathrm{I}_0$ and $\mathrm{I}_1$ satisfy that:
\begin{enumerate}[\rm(i)]
\item $\mathrm{I}_0(-\mu)=\mathrm{I}_0(\mu), \mathrm{I}_1(-\mu)=-\mathrm{I}_1(\mu)$;
\item $\mathrm{I}_0'(\mu)=\mathrm{I}_1(\mu)$;
\item $\mathrm{I}_0(\mu)>\mathrm{I}_1(\mu)>0$ for all $x>0$;
\item $\dfrac{\mathrm{I}_1(\mu)}{\mathrm{I}_0(\mu)}$ is strictly increasing, with $\displaystyle\lim_{\mu\to+\infty} \frac{\mathrm{I}_1(\mu)}{\mathrm{I}_0(\mu)}=1.$
\end{enumerate}
\end{lemma}
\begin{proof}
The facts (i)--(iii) are standard properties of Bessel functions, so we refer the readers to the citations. We briefly verify fact (iv) for completeness. Taking the derivative, we get
\[\qty(\frac{\mathrm{I}_1(\mu)}{\mathrm{I}_0(\mu)})'=\frac{\mathrm{I}_1'\mathrm{I}_0-\mathrm{I}_1^2}{\mathrm{I}_0^2}=\frac{1}{4\pi^2 \mathrm{I}_0^2} \left[ \int_0^{2\pi} \cos^2\theta\e^{\mu\cos\theta}\d\theta\int_0^{2\pi} \e^{\mu\cos\theta}\d\theta - \left( \int_0^{2\pi} \cos\theta\e^{\mu\cos\theta}\d\theta \right)^2 \right],\]
which is strictly positive by the Cauchy-Schwarz inequality. When $\mu\to+\infty$, we quote the following well-known asymptotic expansion of the Bessel functions when $z\to+\infty$ \cite{abramowitz_handbook_2013}:
\begin{equation} \label{In-asymp}
    \mathrm{I}_\nu(z) \sim \frac{\e^z}{\sqrt{2\pi z}} \left[ 1-\frac{4\nu^2-1}{8z} + \frac{(4\nu^2-1)(4\nu^2-9)}{2!(8z)^2} - \frac{(4\nu^2-1)(4\nu^2-9)(4\nu^2-25)}{3!(8z)^3}+\cdots \right].
\end{equation}
Note that an asymptotic series is not necessarily convergent, but its truncation to $O(x^{-N})$ terms has the precision of $O(x^{-(N+1)})$, like a Taylor polynomial. 

Hence, we have that
\begin{equation}\label{bessel-asymp}
\begin{aligned}
    \mathrm{I}_0(\mu) &= \frac{\e^\mu}{\sqrt{2\pi \mu}} \left[ 1+\frac{1}{8\mu} + \frac{9}{128 \mu^2} +\frac{75}{1024 \mu^3}+ O(\mu^{-4}) \right], \\
    \mathrm{I}_1(\mu) &= \frac{\e^\mu}{\sqrt{2\pi \mu}} \left[ 1-\frac{3}{8\mu} - \frac{15}{128 \mu^2} - \frac{105}{1024 \mu^3}+ O(\mu^{-4}) \right],
\end{aligned}
\end{equation}
which implies $\mathrm{I}_0(\mu), \mathrm{I}_1(\mu) \sim \frac{\e^\mu}{\sqrt{2\pi\mu}}$ as $\mu\to+\infty$, so $\frac{\mathrm{I}_1}{\mathrm{I}_0}\to 1$.
\end{proof}

By Lemma \ref{lem:bessel}(iv), $q(\mu)$ defined by \eqref{q2d-eq-I1} traverses its range $(-\frac12,\frac12)$ as $\mu$ traverses $\BbbR$, leading to the existence of the Bingham closure in 2D. The reader may refer to \cite{weady_fast_2022} for a more detailed proof.
\begin{theorem} \label{thm:bing-cl-2d}
For all $\vb Q\in\mcalQ_p$ with $d=2$, there exists a unique $\vb B\in\mcalS_0$ such that $\vb Q=\nabla_0(\ln Z(\vb B)).$
\end{theorem}

\subsection{Entropy decomposition}

We compute from \eqref{ons-fe-hom} that the molecular entropy $S[f]$ with $f_{\vb Q}$ substituted equals
\begin{equation} \label{SQ-2d}
    S(\vb Q) = \vb B:\vb Q - \ln Z(\vb B),
\end{equation}
where $\vb B$ satisfies $\vb Q=\vb Q(\vb B)$. Therefore, $S(\vb Q)$ is the convex dual \cite{boyd_convex_2004} of $\ln Z(\vb B)$, which implies that $\nabla_0 S(\vb Q)=\vb B$.
$S(\vb Q)$ also depends solely on the eigenvalues of $\vb Q$. Using the parametrization \eqref{B-Q-diag2}, we rewrite it as (note that $\vb B:\vb Q=2\mu q$)
\begin{equation} \label{S-q-2d}
    S(q) = 2\mu q - \ln \mathrm{I}_0(\mu).
\end{equation}

Then, Theorem \ref{thm:SQ-asymp} in 2D takes the form of the following statement.
\begin{proposition} \label{prop:S2d-asymp}
Let $S(q)$ be given by \eqref{S-q-2d}, and
\begin{equation} \label{Shat-q-2d}
    \hat S(q)=-\frac12 \ln\left( \frac12+q \right)\left( \frac12-q \right)
\end{equation}
the leading singular term of Theorem \ref{thm:SQ-asymp} in terms of the eigenvalue $q$. Then, $\Delta S(q)=S(q)-\hat S(q)$ is a Lipschitz continuous function, whose derivative with respect to $q\in(-\frac12,\frac12)$ is uniformly bounded.
\end{proposition}
\begin{proof}
We show that the derivative of $S-\hat S$ is uniformly bounded.

By the convex duality $\vb B=\nabla_0 S(\vb Q)$, we derive with the same argument for \eqref{q2d-eq-I1} that
\[S'(q)=2\mu.\]
On the other hand,
\[\hat S'(q)=-\frac12\left(\frac{1}{\frac12+q}-\frac{1}{\frac12-q}\right).\]
Both derivatives can be explicitly computed from $\mu$. As $\mu\in\BbbR$ and $q\in(-\frac12,\frac12)$ are mapped one-to-one onto each other, it suffices to verify the boundedness of $(S-\hat S)'(q)$ with respect to all $\mu\in\BbbR$, which leads us to asserting that
\begin{equation}\label{DdeltaS-2d}
    g(\mu)=2\mu+\frac12\left(\frac{1}{\frac12+q(\mu)}-\frac{1}{\frac12-q(\mu)}\right)
    =2\mu + \frac{\mathrm{I}_0(\mu)}{\mathrm{I}_0(\mu)+\mathrm{I}_1(\mu)} - \frac{\mathrm{I}_0(\mu)}{\mathrm{I}_0(\mu)-\mathrm{I}_1(\mu)}
\end{equation}
is uniformly bounded for $\mu\in\BbbR$. 

By Lemma \ref{lem:bessel}, $g(\mu)$ is an odd function, so we only need to check the boundedness when $\mu\to+\infty$, where the second term $\frac{\mathrm{I}_0(\mu)}{\mathrm{I}_0(\mu)+\mathrm{I}_1(\mu)}$ is uniformly bounded. 
Subtracting the asymptotic expansions \eqref{bessel-asymp} gives
\begin{equation} \label{I0mI1-asymp}
    \mathrm{I}_0(\mu)-\mathrm{I}_1(\mu) = \frac{\e^\mu}{\sqrt{2\pi \mu}}\left[ \frac{1}{2\mu} + \frac{3}{16\mu^2} + \frac{45}{256 \mu^3} + O(\mu^{-4}) \right].
\end{equation}
Substituting them into the expression \eqref{DdeltaS-2d}, we can compute the asymptotic behaviour of $g(\mu)$ at $+\infty$:
\begin{align*}
    g(\mu) &= 2\mu + \frac{1+\frac18 \mu^{-1} + \frac{9}{128} \mu^{-2} + O(\mu^{-3})}{2-\frac14 \mu^{-1}-\frac{3}{64}\mu^{-2} + O(\mu^{-3})} - \frac{1+\frac18 \mu^{-1} + \frac{9}{128} \mu^{-2} + O(\mu^{-3})}{\frac12\mu^{-1} + \frac{3}{16}\mu^{-2} + \frac{45}{256}\mu^{-3} + O(\mu^{-4})} \\
    &=2\mu + \frac12\left[ 1+\frac14\mu^{-1} + \frac18 \mu^{-2} + O(\mu^{-3}) \right] - 2\mu \left[ 1-\frac{1}{4}\mu^{-1} - \frac{3}{16} \mu^{-2} +O(\mu^{-3}) \right]  \\
    &=1 +\frac12 \mu^{-1} + O(\mu^{-2}).
\end{align*}
The unbounded $O(\mu)$ terms cancel out, and result in a uniformly bounded quantity. 
The uniform Lipschitz continuity of $\Delta S(q)=S-\hat S$ then follows naturally. 
\end{proof}

To provide an intuitive interpretation of Proposition \ref{prop:S2d-asymp}, we plot the functions $S(q)$, $\hat S(q)$, and $\Delta S(q)$ in Figure \ref{fig:bing2d}. It shows that the implicit entropy $S(q)$ blows up as $q \to \pm \frac12$, but this singular behaviour is well captured by the explicit quasi-entropy $\hat S(q)$. Meanwhile, the correction term $\Delta S=S-\hat S$ remains bounded over the domain $(-\frac12,\frac12)$, and can therefore be handled efficiently via numerical approximation.

\begin{figure}
\centering
\includegraphics[width=0.45\textwidth]{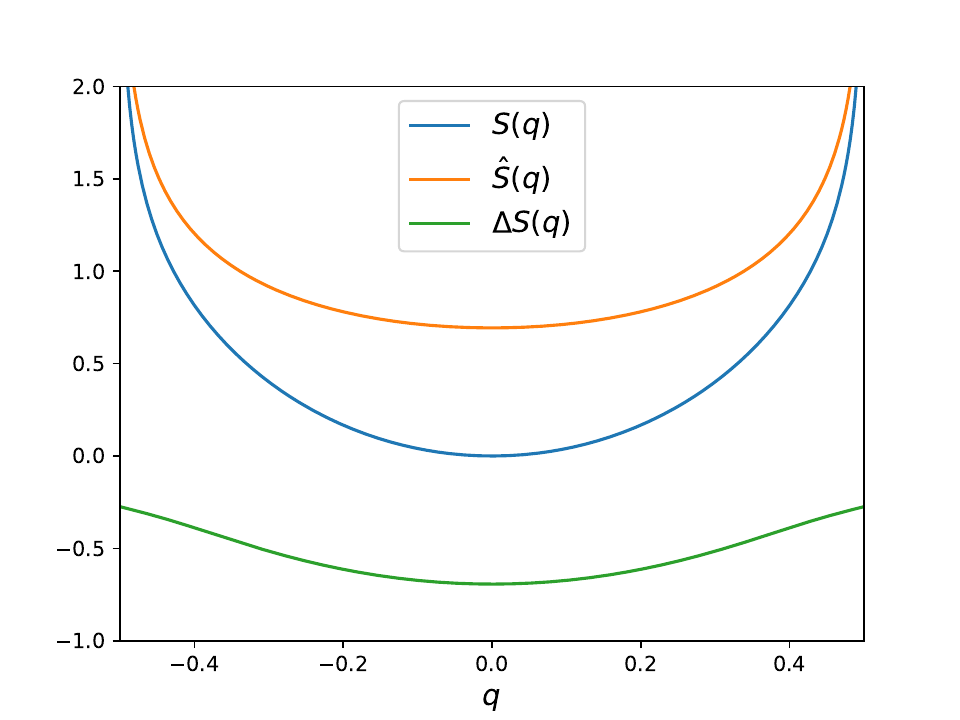}
\caption{The functions $S(q),\hat S(q)$ and $\Delta S(q)$ in the 2D Bingham closure.} \label{fig:bing2d}
\end{figure}

\begin{remark} \label{rmk:weyl-ineq}
Proposition \ref{prop:S2d-asymp} states that $\Delta S$ is Lipschitz continuous with respect to the eigenvalues $\lambda(\vb Q)$. By Weyl's inequality of eigenvalues (see also \cite[Thm.~8.1.4]{golub_matrix_2013}), for any pair of symmetric tensors $\vb Q_1,\vb Q_2$, the difference in their eigenvalues $|\lambda(\vb Q_1)-\lambda(\vb Q_2)|$ is bounded by the matrix norm $|\vb Q_1-\vb Q_2|,$ so $\Delta S$ is also Lipschitz continuous with respect to the tensor inputs, and Theorem \ref{thm:SQ-asymp} is proven.
\end{remark}

\section{Entropy decomposition in dimension 3} \label{sec:bing3d}

We have established the entropy decomposition in the 2D setting, but NLC systems are most often modelled and studied in 3D space. It is therefore natural to extend the same analysis into the 3D case ($d=3$) in this section. For notational convenience, we reuse the symbols $Z,S,\mu,q,$ etc., now interpreted in the 3D setting.

\subsection{Bingham distribution in 3D}

First, we briefly recall the basic properties of the 3D Bingham distribution.
Given $\vb Q\in\mcalQ_p$, we seek for a 3D Bingham distribution of the form similar to \eqref{bingham-2d}
\begin{equation} \label{bingham}
    f(\vb{m})=\frac{1}{4\pi Z} \e^{\vb B:\vb{mm}},\ Z=\frac{1}{4\pi}\int_{\BbbS^2} \e^{\vb B:\vb{mm}}\d\vb m,
\end{equation}
where $\vb B\in\mcalS_0$, so that $f_{\vb Q}$ minimizes the entropy $S[f]$ subject to $\<\vb{mm}-\frac1d\vb I\>_{f_{\vb Q}}=\vb Q.$
By analysing the function $\ln Z$ in 3D, we get that all basic properties listed in Lemma \ref{lem:Z2d} still hold, with the only modification of changing $d$ from 2 to 3. The proofs are identical, so they are omitted. 

Same as the 2D case, $Z$ is rotationally invariant and depends solely on the eigenvalues, and $\vb Q,\vb B$ share the same eigenframe. We can assume w.l.o.g.~that they have a diagonal form:
\begin{equation}
    \vb B=\diag(\mu)=\mqty[\mu_1\\&\mu_2\\&&\mu_3],\ 
    \vb Q=\diag(q)=\mqty[q_1\\&q_2\\&&q_3].
\end{equation}
We rewrite the function $Z$ as determined by the eigenvalues with the same symbol:
\begin{equation} \label{Z-mu}
    Z(\mu)=\frac{1}{4\pi} \int_{\BbbS^2} \exp\left( \sum_{i=1}^3 \mu_i m_i^2 \right)\d\vb m.
\end{equation}
Note that $Z(\mu)$ is well-defined for all $\mu\in\BbbR^3$.
Similarly, $q$ is the projected gradient of $\ln Z(\mu)$ onto the plane $\{x_1+x_2+x_3=0\}$, which we still denote by $\nabla_0$. The components of $q$ are
\begin{equation} \label{qi-def}
    q_i=\<m_i^2-\frac13\>_f=\frac{1}{4\pi Z} \int_{\BbbS^2} \qty(m_i^2-\frac13) \exp\left( \sum_{i=1}^3 \mu_i m_i^2 \right)\d\vb m,\ i=1,2,3.
\end{equation}
The physical condition $\vb Q\in\mcalQ_{p}$ is now expressed in terms of $q$ as
\begin{equation} \label{eigq-phy}
    q\in \Delta_0=\left\{x\in\BbbR^3: x_1+x_2+x_3=0, x_i\in\left(-\frac13,\frac23\right)\right\},
\end{equation}
which can be visualized as an equilateral triangle $\Delta_0$ in $\BbbR^3$ whose vertices are $[\frac23,-\frac13,-\frac13]^T$ and its permutations.


The existence and uniqueness of the Bingham closure in 3D can also be established similarly to Theorem \ref{thm:bing-cl-2d}. This result is well-known and has been proven independently many times \cite{ball_nematic_2010,li_local_2015}.
\begin{theorem}\label{thm:bing-cl}
For all $\vb Q\in\mcalQ_p$ with $d=3$, there exists a unique $\vb B\in\mcalS_0$ such that $\vb Q=\nabla_0(\ln Z(\vb B)).$
\end{theorem}


\subsection{Entropy decomposition}

Then, we discuss Theorem \ref{thm:SQ-asymp} in the 3D case, which involves more complicated asymptotic analysis.

Substituting $f_{\vb Q}$ into the molecular entropy $S[f]$, we find that the 3D entropy $S(\vb Q)$ has the same form as \eqref{SQ-2d}. Therefore, it is still the convex conjugate of $\ln Z(\vb B)$, and $\vb B=\nabla_0 S(\vb Q)$. In terms of the eigenvalues $\mu=\lambda(\vb B),q=\lambda(\vb Q)$, we have that
\begin{equation} \label{S-q}
    S(q)=q\cdot \mu - \ln Z(\mu)
\end{equation}
and that $\nabla_0 S(q)=\mu$.
Here, $\mu$ is determined by the inverse mapping of $q=q(\mu),$ whose existence and smoothness are guaranteed by Theorem \ref{thm:bing-cl}. $S(q)$ is plotted in Figure \ref{fig:S}.

\begin{figure}[th]
\centering
\includegraphics[width=\textwidth]{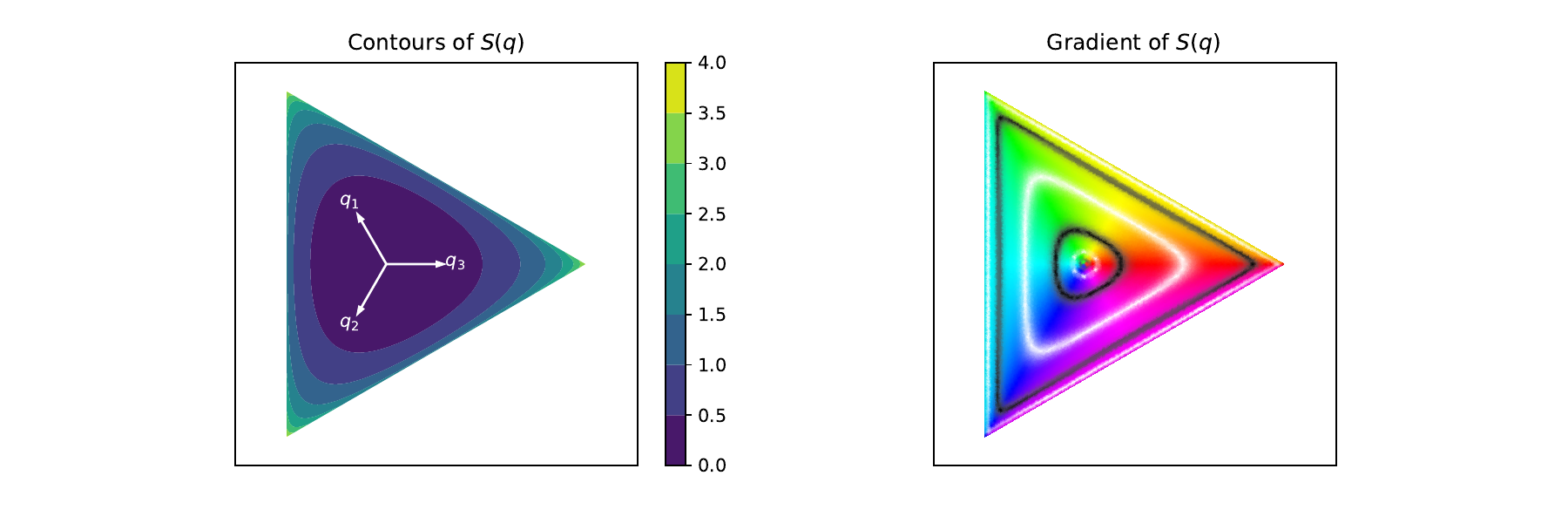}
\caption{Graph of $S(q)$ on its domain $\Delta_0$. Left panel: contour plots of $S$ as a function of $q$; right panel: gradient of $S$. 
The plane $\{q_1+q_2+q_3=0\}\subset\BbbR^3$ is parametrized isometrically by $(\xi,\eta)\in\BbbR^2$ with $\frac{\xi}{\sqrt{6}} [-1,-1,2]^T + \frac{\eta}{\sqrt{2}}[1,-1,0]^T.$
Visualization of $\nabla_0 S$ as a two-dimensional mapping within the $(\xi,\eta)$ plane uses alternating black and white contour lines to represent level curves of the logarithm of the magnitude of the function, and hue (in the order of the RGB colour wheel) for its argument.
Projections of $q_1,q_2,q_3$-axes onto $\{q_1+q_2+q_3=0\}$ are plotted in the left panel.} \label{fig:S}
\end{figure}
 
To simplify notations, we extend the domain of $Z$ \eqref{Z-mu} to all of $\mu\in\BbbR^3$, with the expressions for $q_i$ \eqref{qi-def} unchanged. Denote by $z(\mu)=\nabla[\ln Z(\mu)]$ the regular gradient in $\BbbR^3$. Then, one immediately finds that
\begin{equation} \label{zi-qi}
    z_i=\<m_i^2\>_f=q_i+\frac13,
\end{equation}
i.e.~they differ by a projection.
In addition, there is a ``shift invariance'' in this extension: $\mu+t\bm 1$ and $\mu$ represent essentially the same Bingham distribution for all $t\in\BbbR$ ($\bm 1\triangleq[1,1,1]^T$). Thus, we have that
\begin{equation}\label{Z-mupt-eq}
Z(\mu+t\bm 1)=\e^t Z(\mu),\ z(\mu+t\bm 1)=z(\mu),\ q(\mu+t\bm 1)=q(\mu).
\end{equation}
We shift $\mu$ and let $\mu_1=0\ge\mu_2\ge\mu_3$, which reduces the number of variables by 1 (see e.g.~\cite{kent_asymptotic_1987,luo_fast_2018}).
With the notations established, we present a statement parallel to Proposition \ref{prop:S2d-asymp}. 
We reiterate here that the decomposition of $S(q)$ into an explicit singular leading term $\hat S(q)$ and an implicit Lipschitz continuous correction term $\Delta S(q)$ is aimed at facilitating computation.

\begin{proposition} \label{prop:SmShat-lip}
Work under the notations $Z(\mu)$ \eqref{Z-mu}, $q_i(\mu)$ \eqref{qi-def}, $z_i(\mu)$ \eqref{zi-qi} and
\begin{equation} \label{Shat-q}
    \hat S(q)=-\frac12\ln\det(\vb Q+\frac{\vb I}{3})
    =-\frac12\sum_i\ln\left( q_i+\frac13 \right).
\end{equation}
Then, $\Delta S=S-\hat S$ is a uniformly Lipschitz continuous function, whose gradient with respect to $q\in\Delta_0$ is uniformly bounded.
\end{proposition}
\begin{proof}
To prove Lipschitz continuity of the correction term $\Delta S$, we estimate the uniform bound of the gradient $\nabla_0(\Delta S)$, which equals
\begin{equation} \label{grad-SmShat}
    \nabla_0 (\Delta S) = \left[\mu_1+\frac{1}{2z_1}, \mu_2+\frac{1}{2z_2},\mu_3+\frac{1}{2z_3} \right]^T - c\bm 1,
\end{equation}
where $z_i=q_i+\frac13$, $\bm 1=[1,1,1]^T$ and $c=\frac13(\mu_1+\mu_2+\mu_3)+\frac16(\frac{1}{z_1}+\frac{1}{z_2}+\frac{1}{z_3})$ projects the gradient to $S_0$. The gradient \eqref{grad-SmShat} is fully explicit in $\mu$, which we can use as the primary variable.
By shift invariance \eqref{Z-mupt-eq}, we can assume w.l.o.g.~that $\mu_1=0\ge \mu_2\ge \mu_3$.
The proof is separated into several independent steps relying on several lemmas. To emphasize the main approach, proofs of the lemmas are postponed until the end of this proof.

\textbf{Step 1: Problem reduction.}
Note the elementary fact that for all $a+b+c=0$,
\[(a-b)^2+(b-c)^2+(c-a)^2 = 3(a^2+b^2+c^2).\]
Hence, the norm of \eqref{grad-SmShat}, whose components sum to 0 by definition, is controlled by the differences between its components:
\begin{equation} \label{gradSmShat-le-mui2zimuj2zj}
|\nabla_0 (\Delta S)| \le C \max_{i\neq j} \left| \qty(\mu_i+\frac{1}{2z_i}) - \qty(\mu_j+\frac{1}{2z_j})\right|.
\end{equation}
($z_i=q_i+\frac13$ has been substituted)
If we ensure that one component $\mu_i+\frac{1}{2z_i}$ is bounded, then the problem reduces to estimating the other two components.

We assert that $q_i$ are in the same order as $\mu_i$, which is very intuitive since a larger $\mu_i$ stands for more tendency for the corresponding orientation in the Bingham distribution \eqref{bingham}. 
\begin{lemma} \label{lem:q-mu-order}
    If $\mu_i>\mu_j$, then $q_i>q_j$.
\end{lemma}

By Lemma \ref{lem:q-mu-order}, $z_1$ is the largest among $z_i$ if $\mu_1=0\ge \mu_2\ge \mu_3$, so $\frac13\le z_1\le 1$. Thus, $\mu_1+\frac{1}{2z_1}$ is uniformly bounded for all $\mu_2,\mu_3\le 0$. 
We rewrite $\mu_2=-2x$ and $\mu_3=-2y$, so $x,y\ge 0$.  
Using the spherical coordinate parametrization $m_2=\sin\theta\cos\vph$ and $m_3=\cos\theta$, we get the explicit expression of $z_2=\<m_2^2\>$:
\[z_2(0,-2x,-2y)=\fracd{
    \int_0^{\pi} \e^{-2y \cos^2\theta}\sin\theta\d\theta \int_0^{2\pi} (\sin\theta\cos\vph)^2 \e^{-x \sin^2\theta(1+\cos 2\vph)}\d\vph
}{
    \int_0^{\pi} \e^{-2y\cos^2\theta}\sin\theta\d\theta \int_0^{2\pi} \e^{-x \sin^2\theta(1+\cos 2\vph)}\d\vph
}.\]
After a variable change $u=\sin^2\theta$ and substitution of the Bessel functions $\mathrm{I}_0, \mathrm{I}_1$ (see \eqref{bessel-I0}),
\begin{equation} \label{z2-expr}
    z_2(0,-2x,-2y) = \fracd{
        \int_0^1 u  \e^{-x u} (\mathrm{I}_0(x u)-\mathrm{I}_1(x u)) w_y (u)\d u
    }{
        2\int_0^1 \e^{-x u} \mathrm{I}_0(x u) w_y (u)\d u
    },
\end{equation}
where $w_y(u)=\e^{2yu}(1-u)^{-\frac12}$ is a strictly increasing weight function. 
Proposition \ref{prop:SmShat-lip} then reduces to checking the uniform bound
\begin{equation} \label{sup-1z2x}
    \sup_{x,y\ge 0} \left| \frac{1}{2z_2} - 2x \right| = \sup_{x,y\ge 0} \left| \frac{1-4xz_2}{2z_2} \right|.
\end{equation}
The bound for $\frac{1}{2z_3}-2y$ follows from symmetry.

\textbf{Step 2: Estimation in bounded case.}
We first deal with the trivial case where $x$ is bounded. If $0\le x \le M$, then by the property of Bessel functions (Lemma \ref{lem:bessel}), for all $u\in[0,1]$, $\e^{-xu}[\mathrm{I}_0(xu)-\mathrm{I}_1(xu)]>0$ and $\e^{-xu}\mathrm{I}_0(xu)>0$. Then, there exists positive constants $c(M), C(M)$ such that
\begin{gather*}
    c(M) \le \e^{-xu}(\mathrm{I}_0(xu)-\mathrm{I}_1(xu)) \le C(M), \\
    c(M) \le \e^{-xu} \mathrm{I}_0(xu) \le C(M)
\end{gather*}
for all $x\le M, u\in[0,1]$. Therefore,
\[0< \frac{1}{z_2} \le\frac{2C(M)}{c(M)} \fracd{
    \int_0^1  w_y (u)\d u
}{
    \int_0^1 u w_y (u)\d u
}\le \frac{2C(M)}{c(M)}  \fracd{2\int_{\frac12}^1 w_y(u)\d u}{\frac12 \int_{\frac12}^1 w_y(u)\d u} = 4C'(M). \]
In the second inequality, we use the strict increasing property of $w_y(u)$: $\int_0^{\frac12} w_y(u)\d u<\int_{\frac12}^{1} w_y(u)\d u$. Therefore, for $0\le x\le M$ and all $y\ge 0$, there is a uniform bound
\begin{equation} \label{1z2x-le-CM}
    \sup_{\substack{0\le x\le M \\y\ge 0}} \left| \frac{1}{2z_2}-2x \right| \le C(M).
\end{equation}

\textbf{Step 3: Estimation in unbounded case.}
We proceed to the uniform bound on $\frac{1}{2z_2}-2x$ when $x\to+\infty$. First, we present a coarse $O(x^{-1})$ estimate of the denominator $z_2$.
\begin{lemma}\label{lem:z2-bnd}
Let $\mu=[0,-2x,-2y]^T$. There exists $M>0$ and constants $C_1,C_2$ (dependent on $M$) such that for all $x\ge M, y\ge 0$,
\[\frac{1}{2x+C_1 \sqrt{x}}\le 2z_2(\mu) \le \frac{C_2}{x}.\]
That is, $z_2$ decreases at the same order as $x^{-1}$, and the order constant is independent of $y$. 
\end{lemma}
\begin{remark}
Lemma \ref{lem:z2-bnd} significantly improves the existing estimates of the Bingham moments $z_i$ (such as \cite[Cor.~4]{ball_nematic_2010} or \cite[Prop.~2]{li_local_2015}) by asserting a precise decay rate of $z_2=O((\mu_1-\mu_2)^{-1})$, and, more importantly, its independence on the third eigenvalue $\mu_3$. A symmetric result for $z_3$ and $y=\frac{\mu_1-\mu_3}{2}$ follows obviously: 
\[ \frac{1}{2y+C_1\sqrt{y}} \le 2 z_3(0,-2x,-2y) \le \frac{C_2}{y},\ y\ge M, x\ge 0.\]
On the one hand, if we let $\mu(t)=[2t,-t,-t]^T$ where $t\to+\infty$, then $x,y\to-\infty$, and the upper bounds $z_2\le C_2 x^{-1}, z_3\le C_2 y^{-1}$ assert that $z_2,z_3\to 0$. Hence, the corresponding $q(\mu(t))$ converges to the vertex $V_1=[\frac32,-\frac13,-\frac13]^T$ of $\Delta_0$.
On the other hand, if $\mu$ is bounded and traceless with $\mu_1\ge\mu_2\ge\mu_3$, then the two lesser moments $z_j$ ($j=2,3$) are strictly bounded away from zero by the lower bounds $z_j \ge c(\mu_1-\mu_j)^{-1}$, where $c$ is some constant related to $M$ and $C_1$. Correspondingly, $q=z-\frac13\bm 1$ is strictly bounded away from the edges of $\Delta_0$ by the same amount.
\end{remark}

Then, we present a uniform bound on the numerator $1-4xz_2$, which is also $O(x^{-1})$.
\begin{lemma} \label{lem:1m4xz2-bnd}
Let $\mu=[0,-2x,-2y]^T$. There exists $M>0$ and a constant $C$ (dependent on $M$) such that for all $x\ge M,y\ge 0$
\[|1-4x z_2| \le \frac{C_3}{x}. \]
\end{lemma}
Combining Lemma \ref{lem:z2-bnd} and \ref{lem:1m4xz2-bnd}, we find that for all $x\ge M,y\ge 0$,
\[\left| \frac{1}{2z_2}-2x \right| \le  \frac{ C_3(2x+C_1\sqrt x)}{x} < C \]
is uniformly bounded. Combining it with the estimate for the bounded case \eqref{1z2x-le-CM}, we obtain \eqref{sup-1z2x}. Finally, by \eqref{gradSmShat-le-mui2zimuj2zj}, we conclude that the gradient $\nabla_0(S-\hat S)$ \eqref{grad-SmShat} is uniformly bounded for all $q\in\Delta_0$, and establish Proposition \ref{prop:SmShat-lip}.
\end{proof}

By Proposition \ref{prop:SmShat-lip} and the same argument using Weyl's inequality as Remark \ref{rmk:weyl-ineq}, $\Delta S$ is also uniformly Lipschitz continuous with respect to tensor inputs, so the 3D version of Theorem \ref{thm:SQ-asymp} also holds.

\begin{remark}
The log-determinant function \eqref{Shat-q} highly resembles the \textit{quasi-entropy} first proposed by \cite{xu_quasi-entropy_2022}, so we will refer to it by that name. The original quasi-entropy was constructed phenomenologically, but yielded convincing results in numerical tests.
Our Proposition \ref{prop:SmShat-lip} explains further that this quasi-entropy (slightly different from the original) correctly models the blow-up behaviour of $S(q)$ near the boundary up to a Lipschitz continuous perturbation, highlighting the physical rationale of this surrogate function.
\end{remark}


\subsection{Proofs of lemmas}

First, we present a simple and straightforward proof of Lemma \ref{lem:q-mu-order} about the ordering of eigenvalues.
\begin{proof}[Proof of Lemma \ref{lem:q-mu-order}]
W.l.o.g., we let $\mu_1>\mu_2$. We compute \eqref{qi-def} with the spherical coordinates $m_1=\sin\theta\cos\vph,m_2=\sin\theta\sin\vph, m_3=\cos\theta$: \small
\begin{align*}
    &\quad q_1-q_2\\
    &=\fracd{
        \int_0^\pi\e^{\mu_3\cos^2\theta}\sin\theta\d\theta \int_0^{2\pi} \e^{\sin^2\theta(\mu_1\cos^2\vph+\mu_2\sin^2\vph)}\sin^2\theta(\cos^2\vph-\sin^2\vph)\d\vph
    }{
        \int_0^\pi\e^{\mu_3\cos^2\theta}\sin\theta\d\theta \int_0^{2\pi} \e^{\sin^2\theta(\mu_1\cos^2\vph+\mu_2\sin^2\vph)}\d\vph
    }\\
    &=
    \fracd{
        \int_0^\pi \e^{\mu_3\cos^2\theta+\mu_2\sin^2\theta}\sin\theta\d\theta \int_0^{\pi/4} [\e^{(\mu_1-\mu_2)\sin^2\theta\cos^2\vph}-\e^{(\mu_1-\mu_2)\sin^2\theta\sin^2\vph}]\sin^2\theta\cos 2\vph\d\vph
    }{
        \int_0^\pi\e^{\mu_3\cos^2\theta+\mu_2\sin^2\theta}\sin\theta\d\theta \int_0^{\pi/4} [\e^{(\mu_1-\mu_2)\sin^2\theta\cos^2\vph} +\e^{(\mu_1-\mu_2)\sin^2\theta\sin^2\vph}]\d\vph
    }>0.
\end{align*}\normalsize
\end{proof}

Then, we prove Lemma \ref{lem:z2-bnd}, which states the precise $O(x^{-1})$ decay rate of $z_2$.
\begin{proof}[Proof of Lemma \ref{lem:z2-bnd}]
For the upper bound, we use Cauchy's mean value theorem on the fraction in \eqref{z2-expr}. There exists $\xi\in(0,1)$ such that
\[2z_2 x = \frac{x\xi(\mathrm{I}_0(x\xi)-(\mathrm{I}_1(x\xi)))}{\mathrm{I}_0(x\xi)} \triangleq \phi(x\xi), \]
where $\phi(z)=\frac{z(\mathrm{I}_0(z)-\mathrm{I}_1(z))}{\mathrm{I}_0(z)}.$
The function $\phi(z)$ is obviously bounded for bounded $z$; moreover, the asymptotic expansions \eqref{bessel-asymp} and \eqref{I0mI1-asymp} show that
\[\mathrm{I}_0(z)= \frac{\e^z}{\sqrt{2\pi z}}(1+O(z^{-1})),\ 
\mathrm{I}_0(z)-\mathrm{I}_1(z)= \frac{\e^z}{\sqrt{2\pi z}}\qty(\frac{1}{2z}+O(z^{-2})),\]
so $\phi(z)$ has a finite limit when $z\to\infty$:
\[\lim_{z\to\infty} \phi(z) = \lim_{z\to\infty}\frac{z\cdot \frac1{2z} + O(z^{-1})}{ 1 + O(z^{-1})} = \frac12.\]
Therefore, $2 z_2 x$ has a uniform upper bound $C_2$, i.e. $2z_2\le \frac{C_2}{x}$.

For the lower bound, we need more details about the asymptotic expansions \eqref{bessel-asymp} and \eqref{I0mI1-asymp}. For any $\ve>0$, there exists $M$ such that when $z\ge M$, the following estimates hold.
\begin{equation} \label{emxI0x-asymp-bnd}
\begin{aligned}
    \frac{1}{\sqrt{2\pi z}}\left[ 1+\left( \frac18-\ve \right)z^{-1} \right]& <\e^{-z} \mathrm{I}_0(z) < \frac{1}{\sqrt{2\pi z}}\left[ 1+\left( \frac18+\ve \right)z^{-1} \right], \\
    \frac{1}{\sqrt{2\pi z}} \left[ \frac12 z^{-1} + \left( \frac{3}{16}-\ve  \right) z^{-2} \right] &< \e^{-z}(\mathrm{I}_0(z)-\mathrm{I}_1(z)) < \frac{1}{\sqrt{2\pi z}} \left[ \frac12 z^{-1} + \left( \frac{3}{16}+\ve  \right) z^{-2} \right].
\end{aligned}
\end{equation}
Then, if $x\ge M$, we use the upper and lower bounds above (which apply on $u\in[\frac{M}{x},1]$) to estimate that
\begin{align*}
    2z_2 x &> \fracd{
        \int_{M/x}^1 xu \qty[\frac12 (xu)^{-\frac32} +\qty(\frac{3}{16}-\ve)(xu)^{-\frac52}] w_y(u)\d u 
    }{
        C(M) \int_0^{M/x} w_y(u)\d u + \int_{M/x}^1 \qty[(xu)^{-\frac12} + \qty(\frac18+\ve)(xu)^{-\frac32}] w_y(u)\d u} \\
    &\triangleq \frac{\frac12 J_0 + \qty(\frac{3}{16}-\ve)x^{-1}J_1}{x^{\frac12} R+J_0 + \qty(\frac18+\ve)x^{-1}J_1},
\end{align*}
where $C(M)$ is the upper bound of $\e^{-z}\mathrm{I}_0(z)$ on $z\in[0,M]$ and the abbreviations are
\begin{equation}  \label{R-J0-J1}
    R = C(M)\int_0^{M/x} w_y(u)\d u, \
    J_0 =  \int_{M/x}^1 u^{-\frac12}w_y(u)\d u, \
    J_1 =  \int_{M/x}^1 u^{-\frac32}w_y(u)\d u.
\end{equation}

We provide bounds for $R,J_0,J_1$ when $x\ge 2M$.
First, by the monotonicity of $w_y(u)=\e^{2yu}(1-u)^{-\frac12}$, we have that
\[R \le C(M) \cdot \frac{M}{x}\cdot w_y\qty(\frac{M}{x}) \le MC(M) x^{-1} w_y\qty(\frac12),\]
and that 
\[J_0 \ge \int_{\frac12}^1 u^{-\frac12} w_y\qty(\frac12)\d u \ge \frac12 w_y\qty(\frac12),\]
provided $x\ge 2M$. Therefore, 
\begin{equation} \label{R-J0-x--1}
    \frac{R}{J_0} \le c_1 x^{-1},
\end{equation}
where $c_1$ depends on $M$ only.
The estimation of $J_1$ is more technical since the integration of $x^{-\frac32}$ on $[0,1]$ is divergent. Taking the derivative of $w_y(u)=\e^{2yu}(1-u)^{-\frac12}$ with respect to $y$, we get that
\[\frac{\pt}{\pt y} \frac{J_1}{J_0}
=\frac{2}{J_0^2} \left[ \left( \int_{M/x}^1 u^{-\frac12} w_y(u)\d u\right)^2 - \int_{M/x}^1 u^{-\frac32} w_y(u)\d u \int_{M/x}^1 u^{\frac12} w_y(u)\d u \right] \le 0\]
by Cauchy-Schwarz inequality. Hence, $\frac{J_1}{J_0}$ is a nonnegative and decreasing function of $y$, and is bounded by its value at $y=0$:
\[ \frac{J_1}{J_0}\le \frac{J_1}{J_0}\Big|_{y=0}
= \fracd{\int_{M/x}^1 u^{-\frac32} (1-u)^{-\frac12}\d u}{
    \int_{M/x}^1 u^{-\frac12} (1-u)^{-\frac12}\d u
}\]
For $x\ge 2M$, the denominator is bounded below by a constant (equal to $\frac12\mathrm{B}(\frac12,\frac12)=\frac{\pi}{2}$), while the numerator is bounded above by
\[\int_{M/x}^1 u^{-\frac32} (1-u)^{-\frac12}\d u
= 2 \sqrt{\frac{1-M/x}{M/x}} \le c' x^{\frac12}.\]
Therefore, we have the estimate
\begin{equation} \label{J1-J0-sqrtx}
    \frac{J_1}{J_0} \le c_2 x^{\frac12},
\end{equation}
where $c_2$ depends on $M$ only.
Combining \eqref{R-J0-x--1} and \eqref{J1-J0-sqrtx} into the expression of $z_2 x$, we find that
\begin{align*}
\frac{1}{2z_2 x} &< 2 \frac{x^{\frac12}R +J_0+ (\frac18+\ve) x^{-1}J_1}{J_0 + x^{-1} ( \frac38-2\ve ) J_1}
\le 2 \qty[ 1 + x^{\frac12} \frac{R}{J_0} + \qty(\frac18+\ve) x^{-1} \frac{J_1}{J_0}] \\
& \le 2 + C_1 x^{-\frac12},
\end{align*}
which leads to the desired lower bound estimate $2z_2 \ge \frac{1}{2x + C_1\sqrt{x}}$ ($x\ge 2M$). The proof of Lemma \ref{lem:z2-bnd} is hereby complete.
\end{proof}

Finally, we prove Lemma \ref{lem:1m4xz2-bnd} about the same $O(x^{-1})$ order of the numerator $1-4xz_2$. This proof is the most technical.
\begin{proof}[Proof of Lemma \ref{lem:1m4xz2-bnd}]
By \eqref{z2-expr},
\begin{equation} \label{1m4xz2-expr}
1-4xz_2 = \fracd{\int_0^1 \psi_1(xu)w_y(u)\d u}{\int_0^1 \psi(xu) w_y(u)\d u},
\end{equation}
where we use the abbreviations
\begin{align*}
    \psi(z) &= \e^{-z}\mathrm{I}_0(z),\\
    \psi_1(z) &= \e^{-z}[\mathrm{I}_0(z)-2z(\mathrm{I}_0(z)-\mathrm{I}_1(z))].
\end{align*}
Lemma \ref{lem:1m4xz2-bnd} then follows from the following claims, which hold for all $x \ge M$ with absolute constants $c_1,c_2,c_3$ dependent only on $M$.
\begin{align}
    \text{Control of denominator: }&\int_0^1 \psi(xu)w_y(u)\d u\ge c_1 x^{-\frac12} \int_{\frac12}^1 w_y(u)\d u, \label{int-psi-lb} \\
    \text{Control of numerator (away from 0): }& \left|\int_{\frac12}^1 \psi_1(xu)w_y(u)\d u \right| \le c_2 x^{-\frac32} \int_{\frac12}^1 w_y(u)\d u, \label{int-psi1-ub-far} \\
    \text{Control of numerator (close to 0): }& \left| \int_{0}^{\frac12} \psi_1(xu)w_y(u)\d u \right| \le c_2 x^{-\frac32} \int_{\frac12}^1 w_y(u)\d u. \label{int-psi1-ub-near}
\end{align}

We verify \eqref{int-psi-lb} and \eqref{int-psi1-ub-far} first. 
Take the $\ve,M$ in \eqref{emxI0x-asymp-bnd}, and then $\psi(xu)\ge c_1 (xu)^{-\frac12} \ge c_1 x^{-\frac12}$ for $x\ge 2M, u\in[\frac12,1]$ and some constant $c_1=c_1(M)$. Thus,
\[\int_0^1 \psi(xu)w_y(u) \d u 
\ge \int_0^{\frac12} \psi(xu)w_y(u) \d u 
\ge  c_1 x^{-\frac12} \int_{\frac12}^1 w_y(u)\d u. \]
Similarly, we use the asymptotic expansion of $\psi_1(z)$ obtained from \eqref{bessel-asymp} and \eqref{I0mI1-asymp}
\begin{equation}
    \psi_1(z) = \frac{1}{\sqrt{2\pi z}}\qty[-\frac{1}{4z} - \frac{9}{32 z^2} + O(z^{-3})]
\end{equation}
to derive that $|\psi_1(xu)|\le c'_2 (xu)^{-\frac32} \le c_2 x^{-\frac32}$ for $x\ge 2M, u\in[\frac12,1]$ and some $c_2=c_2(M)$. Thus,
\[ \left|\int_{\frac12}^1 \psi_1(xu)w_y(u)\d u \right| \le c_2 x^{-\frac32} \int_{\frac12}^1 w_y(u)\d u.\]

Then, we consider \eqref{int-psi1-ub-near}. 
We make the following key observation:
\begin{equation} \label{int-psi1-0}
    \int_0^\infty \psi_1(z)\d z =0.
\end{equation}
\begin{proof}[Proof of \eqref{int-psi1-0}]
By the asymptotic expansions \eqref{bessel-asymp} and \eqref{I0mI1-asymp}, $\psi_1(z)=O(z^{-\frac32})$ when $z\to+\infty$, so the integral is convergent in $L^1$. By the dominated convergence theorem,
\begin{equation} \label{int-psi1-domconv}
    \int_0^\infty \psi_1(z)\d z = \lim_{\lambda\to 1^+} \int_0^\infty \e^{-(\lambda-1) z} \psi_1(z)\d z.
\end{equation}
We quote the following Laplacian transform of $\mathrm{I}_0$ from \cite{abramowitz_handbook_2013}:
\begin{equation}
\int_0^\infty \e^{-\lambda z} \mathrm{I}_0(z) \d z = \frac{1}{\sqrt{\lambda^2-1}},\ 
\int_0^\infty \e^{-\lambda z} \mathrm{I}_1(z)\d z= \frac{\lambda}{\sqrt{\lambda^2-1}}-1,\ \lambda>1
\end{equation}
Differentiating them with respect to $\lambda$, we get
\begin{equation}
    \int_0^\infty z \e^{-\lambda z} \mathrm{I}_0(z)\d z= \frac{\lambda}{(\lambda^2-1)^{\frac32}},\ \int_0^\infty z \e^{-\lambda z} \mathrm{I}_1(z)\d z= \frac{1}{(\lambda^2-1)^{\frac32}}.
\end{equation}
Then, we compute the RHS of \eqref{int-psi1-domconv}. After some simplification, we get
\begin{align*}
    \int_0^\infty\psi_1(z)\d z
    &= \lim_{\lambda\to 1^+} \left[ \int_0^\infty \e^{-\lambda z} \mathrm{I}_0(z) \d z - 2 \int_0^\infty z\e^{-\lambda z} \mathrm{I}_0(z)\d z +2 \int_0^\infty z\e^{-\lambda z} \mathrm{I}_1(z)\d z \right] \\
    &= \lim_{\lambda\to 1^+} \frac{(\lambda-1)^2}{(\lambda^2-1)^{\frac32}} = 0. 
\end{align*}
\end{proof}
Back to estimating \eqref{int-psi1-ub-near}. Let $\Psi_1(z)$ be the primitive function of $\psi_1$:
\begin{equation}
    \Psi_1(z)=\int_0^z \psi_1(z)\d z.
\end{equation}
Then, $\Psi_1(0)=0$, $\lim_{z\to\infty}\Psi_1(z)= 0$, and $|\Psi_1'(z)| = |\psi_1(z)| \le C z^{-\frac32}$ for $z\ge M$. Hence, by the Newton-Leibniz formula we can estimate $\Psi_1$:
\begin{equation} \label{bigpsi1-sqrtx}
    |\Psi_1(z)| \le \int_z^\infty |\psi_1(t)|\d t \le  C' z^{-\frac12},\ z\ge M.
\end{equation}
We rewrite the integral \eqref{int-psi1-ub-near} with integration by parts ($x\ge 2M$)
\begin{align*}
    &\ \int_0^{\frac12} \psi_1(xu)w_y(u)\d u
    =\int_0^{\frac12} \Psi_1'(xu)w_y(u) \d u \\
    &= \frac1x\left[ \Psi_1\qty(\frac{x}{2}) w_y\qty(\frac12) - \qty(\int_0^{M/x}+\int_{M/x}^{\frac12}) \Psi_1(xu) w_y'(u)\d u\right] \\
    &\triangleq \frac1x (K_1+K_2+K_3).
\end{align*}
The next step is to show that $K_1,K_2,K_3$ are all uniformly of order $O(x^{-\frac12})$. 
When $x\ge 2M$,
\begin{equation} \label{K1-is-x--1/2}
\begin{aligned}
    |K_1| &= \left| \Psi_1\qty(\frac{x}{2}) \right| w_y\qty(\frac12) \le C' \qty(\frac{x}{2})^{-\frac12} \cdot 2\int_{\frac12}^1 w_y(u)\d u \\
    &\le c_3' x^{-\frac12}\int_{\frac12}^1 w_y(u)\d u.
\end{aligned}
\end{equation}
For the first inequality, we have used both the estimate \eqref{bigpsi1-sqrtx} (applies to $\Psi_1(\frac{x}{2})$ when $x\ge 2M$) and the monotonicity of $w_y(u)=\e^{2yu}(1-u)^{-\frac12}$.
Since $\Psi_1(+\infty)=0$, $\Psi_1$ has a uniform bound
\begin{equation} \label{Psi1-unibnd}
    \sup_{z>0}|\Psi_1(z)|<\infty.
\end{equation}
Applying \eqref{Psi1-unibnd} to $K_2$ and \eqref{bigpsi1-sqrtx} to $K_3$, we get that
\begin{align}
    |K_2| &\le \sup_{z>0}|\Psi_1(z)| \int_0^{M/x} w_y'(u)\d u \le c_3'' x^{-1} \sup_{[0,1/2]} |w_y'(u)| \label{K2-is-x--1/2} \\
    |K_3| &\le C'\int_{M/x}^{\frac12} (xu)^{-\frac12} w_y'(u)\d u \le c_3''' x^{-\frac12} \sup_{[0,1/2]} |w_y'(u)|. \label{K3-is-x--1/2}
\end{align}
Hence, it remains to control $\sup_{[0,\frac12]} |w_y'(u)|$ by the RHS of \eqref{int-psi1-ub-near}. We evaluate that
\[w_y'(u) = \frac{2y\e^{2yu}}{\sqrt{1-u}} + \frac{\e^{2yu}}{2(1-u)^{\frac32}} \le 2\sqrt{2} (y+1)\e^y,\ \forall u\le \frac12.\]
Since any polynomial can be uniformly controlled by exponentials for all $y\ge 0$, there exists absolute constants $C,C'$ such that 
\[(y+1)\e^y \le C \e^{\frac32 y} \le C' \int_{\frac34}^1 \frac{\e^{2yu}}{\sqrt{1-u}}\d u \le C' \int_{\frac12}^1 w_y(u)\d u.\]
With this estimate combined with \eqref{K1-is-x--1/2}, \eqref{K2-is-x--1/2} and \eqref{K3-is-x--1/2}, we conclude that
\begin{equation}
    \left| \int_0^{\frac12} \psi_1(xu)w_y(u)\d u \right|
    \le \frac1x(|K_1|+|K_2|+|K_3|) \le c_3 x^{-\frac32} \int_{\frac12}^1 w_y(u)\d u,
\end{equation}
and \eqref{int-psi1-ub-near} is hereby established. 

By \eqref{int-psi-lb}, \eqref{int-psi1-ub-far} and \eqref{int-psi1-ub-near}, we get
\[1-4xz_2 \le \fracd{(c_2+c_3) x^{-\frac32} \int_{\frac12}^1 w_y(u)\d u}{c_1 x^{-\frac12} \int_{\frac12}^1 w_y(u)\d u} = \frac{C_3}{x},\]
which completes the proof of Lemma \ref{lem:1m4xz2-bnd}.
\end{proof}

\subsection{Dimension reduction: embedding of the 2D model in the 3D model}

In the study of planar NLC devices such as shallow confinements, it is reasonable to reduce a 3D $\vb Q$-tensor model to the 2D model, where the NLC directors tend to lie in a 2D plane, and are homogeneous along the $z$-direction. This approach has also been adopted by the LdG model \cite{canevari_well_2020,golovaty_dimension_2015}.

In this subsection, we prove a similar dimension reduction result for the mean-field $\vb Q$-tensor model with Bingham closure. To distinguish the notations in 2D and 3D, tildes will be attached to functions and variables ($\mu,q,Z,S$, etc.) in 2D. 

\begin{theorem} \label{thm:bing3d-to-2d}
\begin{enumerate}[\rm(i)]
\item Let $q^*=[q_1^*,q_2^*,-\frac13]^T$ be on the interior of an edge of $\pt\Delta_0$ (i.e.~$q_1^*,q_2^*>-\frac13$). Then,
\[\Delta S(q^*) = \Delta\tilde S(\tilde q) - \frac{1+\ln(\pi/2)}{2},\]
where $\tilde q=\frac12(q^*_1-q^*_2).$ Values on the other edges are determined identically by symmetry.
\item At the boundary points $\tilde q=\pm\frac12$, we have that
\[\Delta \tilde S\qty(\pm\frac12) = -\frac{1-\ln(\pi/2)}{2}.\]
Consequently, $\Delta S=-1$ at the vertices of $\Delta_0$.
\end{enumerate}
\end{theorem}
\begin{remark}
Theorem \ref{thm:bing3d-to-2d} shows that the 2D correction function $\Delta\tilde S(\tilde q)=\tilde S(\tilde q)-\tilde{\hat S}(\tilde q)$ is identical to the slices of the 3D correction function $\Delta S$ along the edges of $\Delta_0$.
This implies an embedding relation of the 2D Bingham closure model in to the 3D model as a limiting case: when one eigenvalue $q_i$ tends to $-\frac13$ in 3D (i.e.~$\mu_i$ tends to $-\infty$ and the $i$-th component of the director vanishes), the Bingham distribution degenerates from 3D \eqref{bingham} into 2D \eqref{bingham-2d}, and so does its entropy $S$ after removing the logarithmic singularity in the $i$-th component.
\end{remark}

\begin{proof}[Proof of Theorem \ref{thm:bing3d-to-2d}]
(i) By symmetry, we can assume $q^*_1\ge q^*_2>-\frac13$. As $\Delta S$ is continuous up to the boundary, we only need to check the limit of $\Delta S$ for a special sequence $q_n\to q^*$.
To produce such a sequence, we let $\mu_1=0, \mu_2=-2\beta$ fixed and let $\mu_3\to-\infty$, with $\beta$ determined later on. We quote the following result regarding the limiting behaviour of $Z(\mu)$.
\begin{lemma}[\cite{kent_asymptotic_1987}] \label{lem:Zmu-asymp}
Work under the assumption $0=\mu_1\ge\mu_2\ge\mu_3$, and denote by $\lambda=\mu_2-2\mu_3, \beta=-\mu_2/2.$ If $\beta$ is constant and $\lambda\to+\infty$, then
\begin{equation} \label{Z-asymp-1}
    Z\qty(0,-2\beta,-\beta-\frac{\lambda}{2}) \sim \e^{-\beta} \mathrm{I}_0(\beta) \sqrt{\frac{\pi}{2\lambda}} \sum_{k=0}^\infty \frac{(-\beta)^k (2k-1)!! B_k(\beta)}{\lambda^k k!},
\end{equation}
where $\mathrm{I}_0$ is the modified Bessel function from \eqref{bessel-I0} and 
\begin{equation}
    B_k(\beta) = \frac{1}{2\pi \mathrm{I}_0(\beta)}\int_0^{2\pi} \e^{\beta \cos\theta}\cos^k\theta\d\theta
\end{equation}
are the moments of the von Mises-Fisher distribution \cite{mardia_statistics_1975}. Term-wise differentiation applies to \eqref{Z-asymp-1}.
\end{lemma}
Return to estimating $S-\hat S$. 
We truncate \eqref{Z-asymp-1} to $O(\lambda^{-2})$ relative error:
\[Z=\e^{-\beta}\mathrm{I}_0(\beta) \sqrt{\frac{\pi}{2\lambda}} \left[ 1 - \frac{\beta B_1}{\lambda} + O(\lambda^{-2}) \right].\]
Then, we apply term-wise differentiation with respect to $\lambda,\beta$. By the chain rule,
\[\frac{\pt Z}{\pt \lambda}=-\frac12 \frac{\pt Z}{\pt \mu_3},\ \frac{\pt Z}{\pt\beta}=-2 \frac{\pt Z}{\pt \mu_2}- \frac{\pt Z}{\pt \mu_3},\]
from which we can resolve
\begin{align*}
    \frac{\pt Z}{\pt\mu_3}&=-2\frac{\pt Z}{\pt\lambda}
=\e^{-\beta}\mathrm{I}_0(\beta) \sqrt{\frac{\pi}{2\lambda}} \cdot \frac{1}{\lambda} \left[ 
    1 - \frac{3\beta B_1}{\lambda} + O(\lambda^{-2}) \right], \\
    \frac{\pt Z}{\pt\mu_2} &= -\frac12 \frac{\pt Z}{\pt\beta} + \frac{\pt Z}{\pt\lambda} \\
    &=\e^{-\beta}\mathrm{I}_0(\beta) \sqrt{\frac{\pi}{2\lambda}} \left[ 
        \frac{1-B_1}{2} - \frac{\beta(B_1-B_2)-B_1-1}{2\lambda} + O(\lambda^{-2}) \right].
\end{align*}
We have used the property $\frac{\d}{\d\beta}(\mathrm{I}_0 B_k) = \mathrm{I}_0 B_{k+1}.$
We can then compute the asymptotic expansion of $z_3$ and $z_2$ as follows.
\begin{subequations} \label{B-asymp-1}
\begin{align}
    z_3&=\frac1Z \frac{\pt Z}{\pt\mu_3}
    =\frac{1}{\lambda} \left[ 1 - 2\beta B_1\lambda^{-1} + O(\lambda^{-2})\right], \\
    z_2&=\frac1Z \frac{\pt Z}{\pt\mu_2}
    =\frac{1-B_1}{2} + \frac{\beta(B_2-B_1^2)-B_1-1}{2}\lambda^{-1} + O(\lambda^{-2}).
\end{align} 
Using $z_1+z_2+z_3=1$, we get
\begin{equation}
    z_1=\frac{1+B_1}{2} -\frac{\beta(B_2-B_1^2)-B_1+1}{2}\lambda^{-1} + O(\lambda^{-2}). 
\end{equation}
\end{subequations}
Hence, the limit of the moments $z_1,z_2$ are
\begin{equation} \label{z-limit-1}
    \lim_{\lambda\to\infty} z_{1,2}=\frac12\pm\frac{\mathrm{I}_1(\beta)}{2\mathrm{I}_0(\beta)}=\frac12\pm\tilde q(\beta),
\end{equation}
where $\tilde q$ is the 2D scalar $\vb Q$-tensor with expression \eqref{q2d-eq-I1}. Equivalently, the limit point of $q(\mu)$ when $\lambda\to\infty$ is
\[\qty[ \frac16+\tilde q(\beta), \frac16-\tilde q(\beta), -\frac13 ]^T.\]
Equating this expression with $q^*$, we find that
\[\begin{cases}
    q_1^*=\frac16+\tilde q(\beta) \\
    q_2^*=\frac16-\tilde q(\beta)
\end{cases} \Rightarrow \tilde q(\beta) = \frac{q_1^*-q_2^*}{2}.\]
By Theorem \ref{thm:bing-cl-2d} (Bingham closure in 2D), there exists a unique $\beta\in\BbbR$ such that $\tilde q(\beta) = \frac{q^*_1-q^*_2}{2} \in [0,\frac12).$

By \eqref{zi-qi}, the entropy $S(q)$ can also be expressed in terms of $\mu$ and $z$ as
\begin{equation} \label{S-q-z}
    S(q)=\sum_i \mu_i z_i - \ln Z(\mu),\ \mu\in\BbbR^3,
\end{equation}
while $\hat S(q)=-\frac12\ln\det(\vb Q+\frac{\vb I}{3})$ becomes
\begin{equation} \label{Shat-z}
    \hat S(q)=-\frac12\ln(z_1z_2z_3).
\end{equation}
Substituting \eqref{B-asymp-1} into \eqref{S-q-z} and \eqref{Shat-z}, we compute
\begin{align*}
    \Delta S=S-\hat S
    &=-2\beta z_2 -\left( \beta+\frac{\lambda}{2} \right)z_3 -\ln Z + \frac12\ln (z_1 z_2 z_3) \\
    &=-\qty[\beta(1-B_1)+(\beta^2(B_2-B_1^2)-\beta(B_1+1)) \lambda^{-1}] - \qty(\beta\lambda^{-1} + \frac12 -\beta B_1 \lambda^{-1}) \\
    &\quad -\left[
        -\beta+ \ln(\sqrt{\frac{\pi}{2}} \mathrm{I}_0(\beta)) - \frac12\ln\lambda
        +\ln (1-\beta B_1 \lambda^{-1})
    \right] \\
    &\quad + \frac12\qty[\ln\frac{1+B_1}{2} +\ln (1- \frac{\beta(B_2-B_1^2)-B_1+1}{1+B_1}\lambda^{-1}) ] \\
    &\quad + \frac12\qty[\ln\frac{1-B_1}{2} + \ln(1+\frac{\beta(B_2-B_1^2) - B_1-1}{1-B_1} \lambda^{-1}) ] \\
    &\quad + \frac12\qty[-\ln\lambda + \ln (1-2\beta B_1\lambda^{-1}) ] + O(\lambda^{-2}).
\end{align*}
The logarithms cancel out, and we are left with
\begin{equation}\label{Sq-asymp-1}
    \Delta S= \beta \frac{\mathrm{I}_1(\beta)}{\mathrm{I}_0(\beta)} -\ln \mathrm{I}_0(\beta) + \frac12\ln\frac{1-(\frac{\mathrm{I}_1(\beta)}{\mathrm{I}_0(\beta)})^2}{4} -\frac{1+\ln(\pi/2)}{2}+ O(\lambda^{-1}).
\end{equation}
Recall the definition of the 2D scalar $\vb Q$-tensor mapping $\tilde q(\tilde\mu)$ defined in \eqref{q2d-eq-I1}. Substituting this formula into \eqref{Sq-asymp-1}, we translate this limit into the 2D Bingham distribution:
\begin{align*}
    \Delta S 
    &=2\beta\tilde q(\beta) -\ln \mathrm{I}_0(\beta)+\frac12\ln\qty(\frac12-\tilde q(\beta))\qty(\frac12+\tilde q(\beta))  - \frac{1+\ln(\pi/2)}{2} \\
    &=\Delta\tilde S(\tilde q(\beta)) - \frac{1+\ln(\pi/2)}{2},
\end{align*}
where the 2D Bingham entropy $\tilde S$ and quasi-entropy $\tilde{\hat S}$ are defined in \eqref{S-q-2d} and \eqref{Shat-q-2d} respectively.

Therefore, by the continuity of $\Delta S$, as $\lambda\to\infty$ and $q(\mu) \to q^*=[q_1^*,q_2^*,-\frac13]^T$,
\[\Delta S(q_n) \to \Delta S(q^*)=\Delta \tilde S(\tilde q) - \frac{1+\ln(\pi/2)}{2}.\]
This concludes part (i) of Theorem \ref{thm:bing3d-to-2d}.

(ii) Since $\Delta\tilde S(\tilde q)$ is an even function of $\tilde q$, we focus on $\tilde q>0$. The limit $\tilde q\to(\frac12)^-$ is equivalent to $\tilde q=\tilde q(\beta)$ and $\beta\to+\infty$. We apply the asymptotic expansions \eqref{bessel-asymp} and \eqref{I0mI1-asymp}.
\begin{align*}
    \Delta S^\star(\tilde q)
    &=\frac{\beta \mathrm{I}_1}{\mathrm{I}_0} - \ln \mathrm{I}_0 + \frac12\qty[ \ln\frac{\mathrm{I}_0-\mathrm{I}_1}{2\mathrm{I}_0}+\ln\frac{\mathrm{I}_0+\mathrm{I}_1}{2\mathrm{I}_0}]\\
    &=\frac{\beta \mathrm{I}_1}{\mathrm{I}_0} - \frac32\ln \mathrm{I}_0 + \frac12\ln(\mathrm{I}_0-\mathrm{I}_1) + \frac12\ln\frac{\mathrm{I}_0+\mathrm{I}_1}{2\mathrm{I}_0} - \frac12\ln 2\\
    &=\frac{\beta(1-\frac38\beta^{-1} + O(\beta^{-2}))}{1+\frac18\beta^{-1}+O(\beta^{-2})}
    - \frac32\qty[(\beta-\ln\sqrt{2\pi\beta})+\ln(1+\frac18\beta^{-1}+O(\beta^{-2}))] \\
    &\quad + \frac12\qty[ (\beta-\ln\sqrt{8\pi\beta^3})+\ln(1+\frac38\beta^{-1} + O(\beta^{-2}))] \\
    &\quad + \ln \frac{1-\frac18\beta^{-1}+O(\beta^{-2})}{1-\frac18\beta^{-1}+O(\beta^{-2})}-\frac12\ln 2  \\
    &=-\frac{1-\ln(\pi/2)}{2} + O(\beta^{-1}).
\end{align*}
Letting $\beta\to\infty$, we find that $\Delta\tilde S(\frac12)=-\frac{1-\ln(\pi/2)}{2}$ as desired, which yields part (ii) of Theorem \ref{thm:bing3d-to-2d}. 

Finally, we derive the value of $\Delta S(q)$ at the vertices of $\Delta_0$.
First, we let $q$ converge to an edge point $q^*=[q_1^*,q_2^*,-\frac13]^T$. By part (i), we have that
\[\Delta S(q^*)= \Delta\tilde S(\tilde q)-\frac{1+\ln(\pi/2)}{2},\ \tilde q=\frac{q_1^*-q_2^*}{2}.\]
Then, we let $q^*$ converge to a vertex $V_1=[\frac23,-\frac13,-\frac13]^T$. Then, $\tilde q\to+\frac12$, and by part (ii) $\Delta S(q^*)$ converges to
\[\Delta S(V_1)=-\frac{1-\ln(\pi/2)}{2} -\frac{1+\ln(\pi/2)}{2}=-1. \]
Values at the other two vertices follow by symmetry.
\end{proof}

\section{Neural network approximation to the correction term} \label{sec:nn}

The 2D Bingham closure is relatively simple because it is essentially the inverse of a scalar function involving the well-studied Bessel functions. The difficulties lie in the 3D problem, where the forward mapping is a non-elementary integral. With Theorem \ref{thm:SQ-asymp} established, the troublesome entropy $S(\vb Q)$ is now reduced to a regular correction term that can be effectively approximated. We approximate the 3D Bingham closure with a neural network \cite{shi_neural_2026}.

The forward problem (Bingham distribution \eqref{bingham}) can be reduced to evaluating the values and derivatives of $Z(\mu)$ \eqref{Z-mu} for $\mu\in\BbbR^3$. We compute it through the globally convergent Taylor expansion
\begin{equation} \label{Z-taylor}
\begin{aligned}
    Z&=\frac{1}{4\pi} \sum_{k=0}^\infty \frac{1}{k!} \int_{\BbbS^2} (\mu_1 m_1^2+\mu_2 m_2^2+\mu_3 m_3^2)^k \d\vb m \\
    &= \sum_{\bm\alpha} \frac{1}{\alpha_1!\alpha_2!\alpha_3!} \frac{(2\alpha_1-1)!! (2\alpha_2-1)!! (2\alpha_3-1)!!}{(2\alpha_1+2\alpha_2+2\alpha_3+1)!!}\mu_1^{\alpha_1}\mu_2^{\alpha_2}\mu_3^{\alpha_3},
\end{aligned}
\end{equation}
where $\bm\alpha=(\alpha_1,\alpha_2,\alpha_3)$ is the 3D multi-index. We evaluate this expression for as many terms as possible until the error is less than a prescribed tolerance.
This expansion as well as its term-wise differentiation has produced Figures \ref{fig:S}, where Newton's iteration was used to compute the inverse mapping $\mu=\mu(q)$.

The inverse problem (Bingham closure) is then approximated with a neural network approximation. We improve the algorithm therein by learning the perturbation term $\Delta S(\vb Q)$ instead of the full entropy $S(\vb Q)$, which avoid the singularity near the boundaries.
As the Bingham closure requires computing $\vb B=\nabla_0 S(\vb Q)$, the neural network approximation should be accurate in the $C^1$ norm. The feasibility is guaranteed by Theorem \ref{thm:SQ-asymp}, which ensures the uniform boundedness of the gradient of $\Delta S$, and the Universal Approximation Theorem \cite{hornik_approximation_1991}, which states that a neural network approximation can be arbitrarily close in $C^1$ norm.


We construct a dataset comprising a group of triplets $\{(\vb Q_i, Z_i, \vb B_i)\}_{i\in I}$ prepared by sampling diagonal matrices $\vb B_i=\diag(\mu_i)$ from $\mcalS_0$, where $Z_i=Z(\vb B_i),\vb Q_i=\frac{1}{Z_i} \nabla_0 Z(\vb B_i)$ are computed through \eqref{Z-taylor}, and $S(\vb Q_i),\Delta S(\vb Q_i)$ through \eqref{SQ-2d}.
Then, we apply a feed-forward network denoted by $\Delta S_{\rm NN}(\vb Q_i;\bm\Theta)$, and minimize the loss function
\begin{equation}
    \mcalL(\bm\Theta) = \frac{1}{|I|}\sum_{i\in I} \left[ \left|\Delta S_{\rm NN}(\vb Q_i;\bm\Theta)-(S_i-\hat S(\vb Q_i))\right|^2 - \left|\nabla_0 \Delta S_{\rm NN}(\vb Q_i;\bm\Theta)-(\vb B_i-\nabla_0 \hat S(\vb Q_i))\right|^2 \right].
\end{equation}

The network $\Delta S_{\rm NN}(\vb Q_i;\bm\Theta)$ takes two inputs, maps them through two full-connected linear hidden layers of size 64 with softplus activation functions, and outputs a single scalar.
Note that the input tensor $\vb Q$ is preprocessed into the scalar variables $(\frac12\tr\vb Q^2,\frac13\tr\vb Q^3)$ before entering the neural network, a design aiming at preserving rotational invariance, namely the sole dependence on eigenvalues.
Since the eigenvalue mapping $\vb Q\mapsto\lambda(\vb Q)$ is difficult to compute and non-differentiable, we take the moments $(\frac12\tr\vb Q^2,\frac13\tr\vb Q^3)$ as inputs instead, which corresponds one-to-one to the eigenvalues, and more importantly, allows a convenient differentiation:
\begin{equation}
    \frac{\pt\Delta S_{\rm NN}(\vb Q)}{\pt\vb Q}
    =\frac{\pt\Delta S_{\rm NN}}{\pt(\frac12\tr\vb Q^2)} \vb Q + \frac{\pt\Delta S_{\rm NN}}{\pt(\frac13\tr\vb Q^3)}\left( \vb Q^2-\frac{|\vb Q|^2}{3}\vb I \right).
\end{equation}
This formula automatically preserves rotational invariance relations \eqref{rot-inv} (of the function value) and \eqref{rot-inv-grad} (of the gradient).
We use the ADAM algorithm from PyTorch to train the network for 2000 epochs. We eventually reach test error $\sim 10^{-6}$, and the resulting function $\Delta S_{\rm NN}$ with respect to $q=\lambda(\vb Q)$ is plotted in Figure \ref{fig:S-NN-plot}.

\begin{figure}
    \centering
    \includegraphics[width=0.38\textwidth]{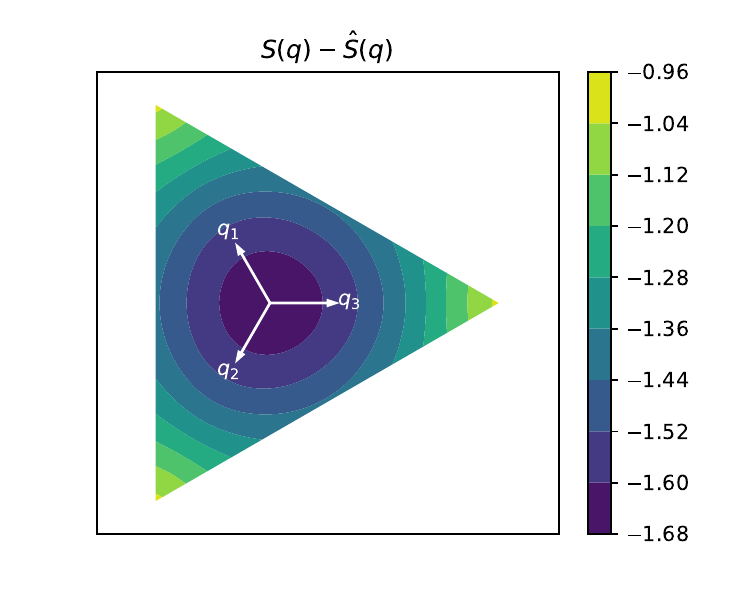}
    \includegraphics[width=0.45\textwidth]{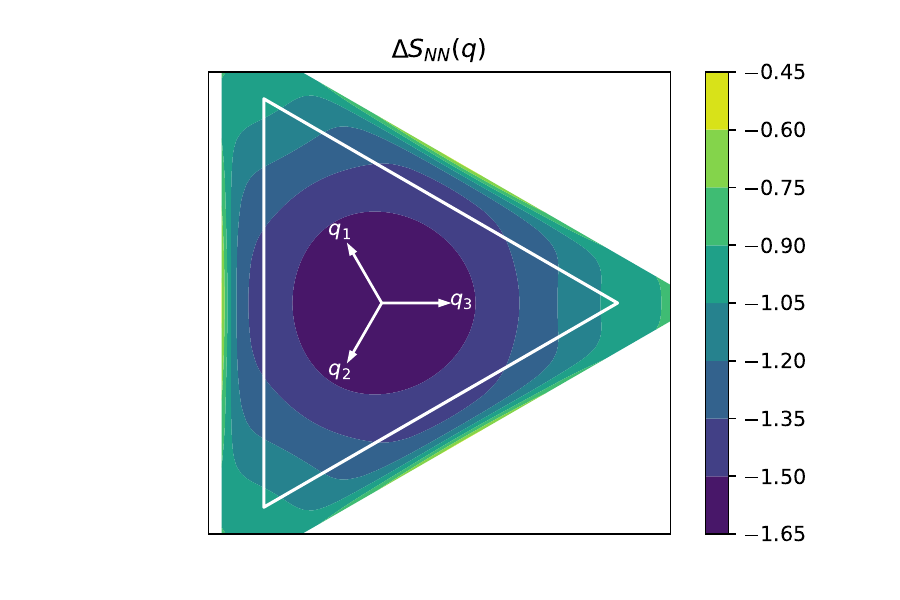}
    \caption{Left panel: Accurate values of $\Delta S$; right panel: neural network approximation $\Delta S_{\rm NN}$. Both are plotted as functions of $q=\lambda(\vb Q)$ in the same way as Figure \ref{fig:S}. $\Delta S_{\rm NN}$ has a nice generalization beyond the physical domain (enclosed by white triangle), which ensures stable computation.}\label{fig:S-NN-plot}
\end{figure}

\section{Numerical experiments} \label{sec:num}

Using the entropy decomposition in Theorem \ref{thm:SQ-asymp} and the neural network approximation introduced in Section~\ref{sec:nn}, we obtain an explicit expression for the entropy term, which is particularly convenient for computation. In this section, we validate the model on (i) the isotropic–nematic phase transition problem and (ii) the liquid crystal droplet problem. Notably, the latter is considered to be more challenging in the context of confinement, since the confining domain can evolve in the droplet setting.

\subsection{Isotropic-nematic phase transition}

First, we consider isotropic-nematic phase transition of the free energy \eqref{ons-fe-hom} under Bingham closure.
For a homogeneous NLC solution with the Maier-Saupe interaction potential \cite{maier_einfache_1959}, the free energy \eqref{ons-fe-hom} is expressed in terms of $\vb Q\in\mcalQ_p$ as a \textit{bulk energy}
\begin{equation} \label{bing-bulk}
    F_b(\vb Q)=S(\vb Q)-\frac{\alpha}{2} |\vb Q|^2,
\end{equation}
where $S(\vb Q)$ \eqref{SQ-2d} replaces the molecular entropy $S[f]$, and the latter term is the exact value of the Maier-Saupe potential (see \cite{ball_nematic_2010} for a derivation). Here, $\alpha>0$ is a nondimensional parameter representing the overall effect of concentration and temperature, and is monotonically decreasing with respect to temperature.

Properties of stationary points of the bulk energy $F_b$ are well-known, especially the isotropic-nematic phase transition. By \cite[App.~A]{fatkullin_critical_2005}, any stationary point of \eqref{bing-bulk} is \textit{uniaxial}, i.e.~having two equal eigenvalues. A uniaxial $\vb Q$-tensor in $\mcalQ_p$ can always be expressed by
\begin{equation} \label{Q-uniax}
    \vb Q=s\qty(\vb{nn}-\frac13\vb I),
\end{equation}
where $\vb{n}\in\BbbS^2$ and $s\in(-\frac12,1)$. If $s>0$, the uniaxial tensor \eqref{Q-uniax} has a single largest eigenvalue, so we regard the direction that leading eigenvector as the principal director, and identify it with the \textit{nematic} phase. If $s=0$, then $\vb Q=\vb B=0$, which represents the uniform distribution, so we identify it with the \textit{isotropic} phase. The isotropic-nematic phase transition is reported in \cite{liu_axial_2005}:
\begin{enumerate}[(i)]
\item The isotropic state $\vb Q=0$ is always a stationary point. When $\alpha>\alpha_*\approx 6.731393$, there exist two types of nematic stationary points with different sets of different eigenvalues.
\item $\vb Q=0$ is stable if and only if $\alpha<7.5$. Otherwise, the isotropic phase loses stability, and one of the nematic phases is the only stable point.
\end{enumerate}

Numerically, we evaluate $F_b$ with the neural network approximation for $\Delta S$:
\[F_b(\vb Q) \approx \hat S(\vb Q) + \Delta S_{\rm NN}(\vb Q) - \frac{\alpha}{2}|\vb Q|^2.\] 
Since the minimizers of $F_b$ are either nematic or isotropic phases, we only need to investigate it along the uniaxial tensors of the form \eqref{Q-uniax}. The approximated energy values are plotted against $s\in(-\frac12,1)$ in Figure \ref{fig:bing-bulk-uniax}. 
One can observe directly that for a smaller $\alpha$ (higher temperature), the unique minimizer is at $s=0$, i.e.~the isotropic state; as $\alpha$ increases from $\alpha_0\approx 6.731393$ to $7.5$, a nematic minimizer (with $s>0$) emerges; for a larger $\alpha>7.5$ (lower temperature), the isotropic phase is unstable, which indicates an isotropic-nematic phase transition. Therefore, the phase transition behaviour of the bulk energy \eqref{bing-bulk} is accurately captured by our approximation.

\begin{figure}[th]
\centering
\includegraphics[width=0.45\textwidth]{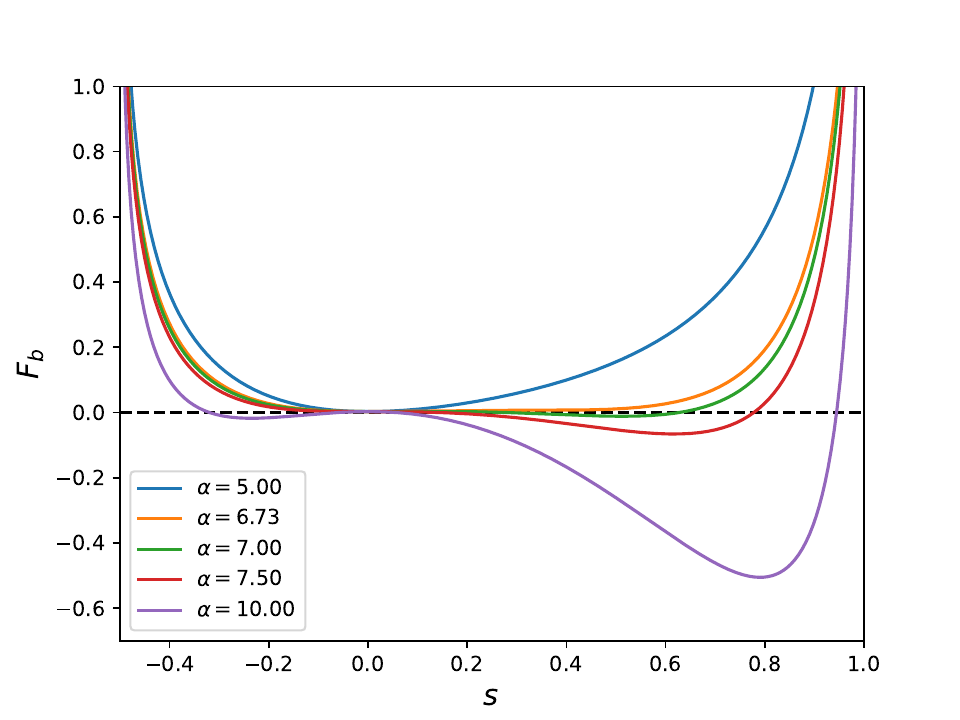}
\caption{The numerically evaluated bulk energy \eqref{bing-bulk} at uniaxial tensors \eqref{Q-uniax} plotted against $s\in(-\frac12,1)$.} \label{fig:bing-bulk-uniax}
\end{figure}


\subsection{Droplet problem}

To further verify that our model can be applied to the confinement problem and even the free-boundary problem \cite{wu_analysing_2026}, we propose the following diffuse-interface energy for a free-boundary NLC droplet based on our entropy decomposition approximation
\begin{equation}\label{bingphf}
\begin{aligned}
    \mathscr{E}[\vb Q,\phi]&=
    \int_\Omega \phi^2\left[ \lambda \left( S(\vb Q)-\frac{\alpha}{2}|\vb Q|^2 - B_{\min}\right) + \frac{1}{2\lambda}|\grad\vb Q|^2 \right]\d x \\
    &\quad + w_p \int_\Omega \qty[
        \ve^{-1}\phi^2 (1-\phi)^2 + \ve|\nabla\phi|^2
        +\omega \ve \left|\qty(\vb Q+\frac{1}{3}\vb I)\nabla\phi \right|^2 ]\d x \\
    &\quad +  \frac{w_v}{2} \int_\Omega \qty[ \kappa |\grad\vb Q|^2 + \kappa^{-1} (1-\phi)^2|\vb Q|^2] \d x
\end{aligned}
\end{equation}
subject to the volume constraint
\begin{equation} \label{vol-con}
    \int_\Omega \phi\d x = V_0.
\end{equation}
Here, $\phi$ is a phase field function such that $\phi\approx 1$ represents the interior of the droplet and $\phi\approx 0$ the isotropic liquid surrounding it.
The energy terms in \eqref{bingphf} are explained as follows.
\begin{itemize}
\item The first integral is the molecular free energy \eqref{ons-fe-hom} under Bingham closure, consisting of the bulk energy $F_b(\vb Q)$ \eqref{bing-bulk} the simplified one-constant elastic energy and $\frac{1}{2\lambda}|\grad\vb Q|^2$ \cite{mottram_introduction_2014}. The parameter $\lambda$ is proportional to the characteristic scale, and $B_{\min}$ is the minimum of \eqref{bing-bulk} that keeps the functional positive. Both terms are masked by $\phi^2$ to restrict the integration to the droplet.
\item The second integral is the interfacial energy, where $w_p$ is a positive weight. The double-well potential $\ve^{-1}\phi^2 (1-\phi)^2$ and the gradient $\ve|\nabla\phi|^2$ constitute the Van der Waals-Cahn-Hilliard energy which approximates the boundary area \cite{cahn_free_1958,modica_gradient_1987,sternberg_effect_1988,van_der_waals_thermodynamic_1979}. The small parameter $\ve$ represents the width of the interface where $\phi$ changes rapidly.
The term $\omega \ve \left|\qty(\vb Q+\frac{1}{3}\vb I)\nabla\phi \right|^2$ is a penalty enforcing the \textit{tangential anchoring condition}, i.e.~the leading eigenvectors of $\vb Q$ should be perpendicular to the boundary normal represented by $\nabla\phi$.
\item The third integral is a penalty enforcing that $\vb Q=0$ in the isotropic region $\phi\approx 0$, where $w_v$ is a positive weight. The gradient term $\kappa|\grad\vb Q|^2$ is added to ensure the overall coercivity of the functional. The small parameter $\kappa$ represents the width of the interface where $\vb Q$ changes.
\end{itemize}
By the requirement of the Bingham entropy $S(\vb Q)$, any well-defined configuration $(\vb Q,\phi)$ must satisfy the physical constraint \eqref{Q-phy}, which overcomes the non-physicality issue occurring in our previous Landau-de~Gennes-based model \cite{majumdar_equilibrium_2010}.

We compute the energy \eqref{bingphf} numerically on the 3D region $\Omega=[0,1]^3$ with periodic boundary conditions. We set the parameters
\[\alpha=8,\ V_0=0.1,\ \ve=0.005,\ \omega=20,\ w_p=1,\ w_v=1,\ \kappa=\sqrt{\ve},\]
and adjust $\lambda=1,2,3,4,5,6$. Note that we choose $\alpha>7.5$ so that a nematic phase is the only stable stationary point of the molecular bulk energy under Bingham closure.
The region is discretized into a finite difference grid of size $N\times N\times N$ where $N=48$, and we evaluate the energy with finite difference, where we utilize the neural network $\Delta S_{\rm NN}(\vb Q)$ to compute the entropy $S(\vb Q)$ under Bingham closure. Then, we run the gradient flow
\[\frac{\pt\vb Q}{\pt t} = -\frac{\delta\mathscr{E}}{\delta\vb Q},\ \frac{\pt\phi}{\pt t} = -\frac{\delta\mathscr{E}}{\delta\phi}+C \]
to minimize the energy. 
The initial guesses are $\vb Q\equiv 0$ and $\phi$ a smooth approximation to a slightly prolate spheroid with volume $V_0$. Finally, we discover the droplet configurations shown in Figure \ref{fig:droplets}. As the scale $\lambda$ increases, biaxial regions (shown in red) are gradually eliminated from the droplet, and the droplet shape is more and more elongated. Eventually, it becomes a nematic tactoid with pointy ends filled with uniaxial tensors, which is consistent with the experimental observations \cite{bawden_liquid_1936,wang_liquid_2018} and the existing numerical results in the diffuse-interface LdG model \cite{wu_analysing_2026,wu_diffuse-interface_2025}. Therefore, the modified diffuse-interface energy \eqref{bingphf} behaves like our previous diffuse-interface LdG model of nematic droplets, but has the additional advantage that the physical admissibility of $\vb Q$ is ensured by the Bingham entropy $S(\vb Q)$.

\begin{figure}
\centering
\includegraphics[width=0.55\textwidth]{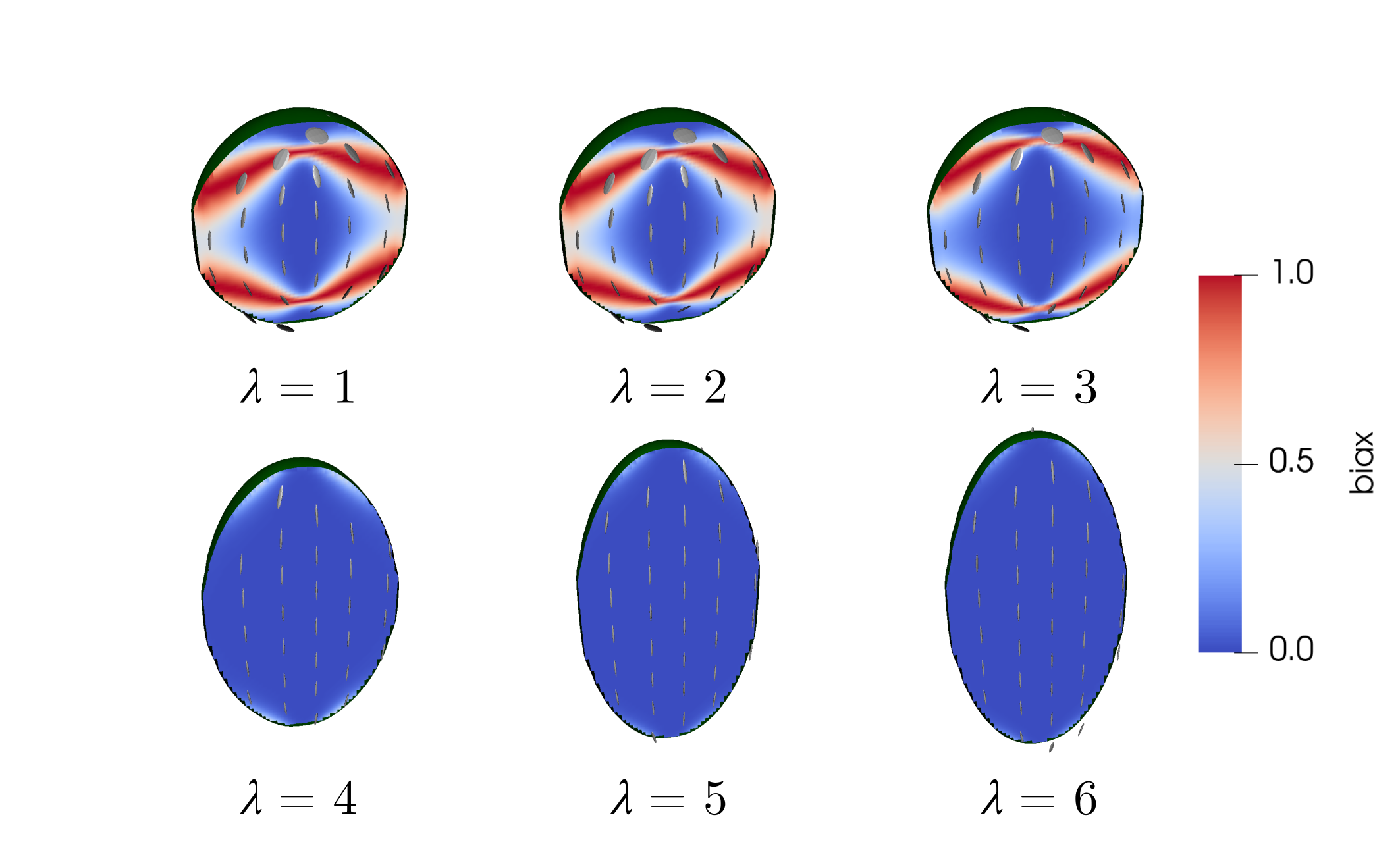}
\caption{Droplets obtained from the gradient flow of \eqref{bingphf}. Boundary of droplet (dark green) is the level set $\phi=0.5$. Colours denote biaxiality $\beta=1-6(\tr\vb Q^3)^2/(\tr\vb Q^2)^3$. White ellipsoids represent $\vb Q$-tensors, with principal axes parallel to the eigenvectors, and their lengths corresponding to eigenvalues in order.} \label{fig:droplets}
\end{figure}

\section{Conclusion}

In this paper, we conduct a detailed analysis of the Bingham closure of Onsager’s molecular free energy, with particular emphasis on the limiting behaviour of the entropy under Bingham closure, $S(\vb Q)$, near the boundaries of the admissible domain, where a logarithmic singularity arises. 
Our principal contribution is the decomposition of the singular entropy into an explicit singular leading term $\hat S$ and an implicit regular correction term $\Delta S$, with special attention in the 3D case (Theorem \ref{thm:SQ-asymp}; Propositions \ref{prop:S2d-asymp} and \ref{prop:SmShat-lip}). This decomposition allows the numerical evaluation of $S(\vb Q)$ to be reduced to the computation of $\Delta S$, thereby effectively circumventing the singular behaviour. Notably, the correction term $\Delta S$ possesses a uniformly bounded gradient that facilitates the application of straightforward and efficient numerical schemes for its approximation.
We further investigate the embedding relationship between the 3D and 2D Bingham closure models (Theorem \ref{thm:bing3d-to-2d}). In particular, we show that the two-dimensional correction term $\Delta\tilde S$ in 2D can be obtained as a slice of its 3D counterpart $\Delta S$ up to an additive constant, which allows us to use the 2D model to reduce computational cost when studying the NLC in 2D or shallow 3D confinements.
Finally, the Lipschitz regularity of $\Delta S$ enables the effective exploitation of the approximation capabilities of neural networks. Numerical experiments confirm that the proposed approximation framework can be successfully applied to the prediction of NLC equilibria and the isotropic–nematic phase transition.

This work suggests several directions for future studies. On the theoretical side, it would be natural to extend the entropy decomposition to higher-dimensional Bingham distributions ($d\ge 4$), which arise in applications such as computer graphics \cite{chen_maximum_2021}. In this setting, an extension of the embedding relation in Theorem \ref{thm:bing3d-to-2d} would also be of interest.
From a computational perspective, the regularity of $\Delta S$ allows the use of higher-order numerical approximations, such as interpolation, which may yield greater accuracy than the neural-network approach considered here.
The proposed framework may also be incorporated into computations of NLC dynamics, where the Bingham closure is known to be effective \cite{feng_closure_1998}.
By combining a relatively simple structure with physically enforced constraints through Bingham closure, the present model offers a promising approach for efficient computation of NLC dynamics, such as the system proposed in \cite{han_microscopic_2015}, which is closely related to the molecular model but has seen limited numerical investigation.

\bibliographystyle{siam}
\bibliography{lc_bib,nn_bib,myworks}

@book{golub_matrix_2013,
	address = {Baltimore},
	edition = {4th},
	series = {Johns {Hopkins} studies in the mathematical sciences},
	title = {Matrix {Computations}},
	isbn = {978-1-4214-0794-4},
	language = {en},
	publisher = {The Johns Hopkins University Press},
	author = {Golub, Gene H. and Van Loan, Charles F.},
	year = {2013},
}

@article{wang_modelling_2021,
	title = {Modelling and computation of liquid crystals},
	volume = {30},
	issn = {0962-4929, 1474-0508},
	url = {https://www.cambridge.org/core/product/identifier/S0962492921000088/type/journal_article},
	doi = {10.1017/S0962492921000088},
	language = {en},
	urldate = {2023-03-09},
	journal = {Acta Numerica},
	author = {Wang, Wei and Zhang, Lei and Zhang, Pingwen},
	month = may,
	year = {2021},
	pages = {765--851},
}

@book{de_gennes_physics_1993,
	address = {Oxford, New York},
	edition = {2nd},
	series = {Oxford science publications},
	title = {The {Physics} of {Liquid} {Crystals}},
	isbn = {978-0-19-852024-5},
	number = {83},
	publisher = {Clarendon Press; Oxford University Press},
	author = {de Gennes, Pierre Gilles and Prost, Jacques},
	year = {1993},
}

@article{han_microscopic_2015,
	title = {From microscopic theory to macroscopic theory: {A} systematic study on modeling for liquid crystals},
	volume = {215},
	issn = {0003-9527, 1432-0673},
	shorttitle = {From {Microscopic} {Theory} to {Macroscopic} {Theory}},
	url = {http://link.springer.com/10.1007/s00205-014-0792-3},
	doi = {10.1007/s00205-014-0792-3},
	language = {en},
	number = {3},
	urldate = {2023-03-09},
	journal = {Archive for Rational Mechanics and Analysis},
	author = {Han, Jiequn and Luo, Yi and Wang, Wei and Zhang, Pingwen and Zhang, Zhifei},
	month = mar,
	year = {2015},
	pages = {741--809},
}

@article{yin_construction_2020,
	title = {Construction of a pathway map on a complicated energy landscape},
	volume = {124},
	issn = {0031-9007, 1079-7114},
	url = {https://journals.aps.org/prl/abstract/10.1103/PhysRevLett.124.090601},
	doi = {10.1103/PhysRevLett.124.090601},
	language = {en},
	number = {9},
	urldate = {2023-03-09},
	journal = {Physical Review Letters},
	author = {Yin, Jianyuan and Wang, Yiwei and Chen, Jeff Z. Y. and Zhang, Pingwen and Zhang, Lei},
	month = mar,
	year = {2020},
	pages = {090601},
}

@article{xu_quasi-entropy_2022,
	title = {Quasi-entropy by log-determinant covariance matrix and application to liquid crystals},
	volume = {435},
	issn = {01672789},
	url = {https://linkinghub.elsevier.com/retrieve/pii/S0167278922000914},
	doi = {10.1016/j.physd.2022.133308},
	language = {en},
	urldate = {2023-03-30},
	journal = {Physica D: Nonlinear Phenomena},
	author = {Xu, Jie},
	month = jul,
	year = {2022},
	pages = {133308},
}

@article{luo_fast_2018,
	title = {A fast algorithm for the moments of {Bingham} distribution},
	volume = {75},
	issn = {0885-7474, 1573-7691},
	url = {http://link.springer.com/10.1007/s10915-017-0589-2},
	doi = {10.1007/s10915-017-0589-2},
	language = {en},
	number = {3},
	urldate = {2023-04-02},
	journal = {Journal of Scientific Computing},
	author = {Luo, Yixiang and Xu, Jie and Zhang, Pingwen},
	month = jun,
	year = {2018},
	pages = {1337--1350},
}

@article{canevari_well_2020,
	title = {The well order reconstruction solution for three-dimensional wells, in the {Landau}-de {Gennes} theory},
	volume = {119},
	issn = {00207462},
	url = {https://linkinghub.elsevier.com/retrieve/pii/S0020746218307212},
	doi = {10.1016/j.ijnonlinmec.2019.103342},
	language = {en},
	urldate = {2023-07-29},
	journal = {International Journal of Non-Linear Mechanics},
	author = {Canevari, Giacomo and Harris, Joseph and Majumdar, Apala and Wang, Yiwei},
	month = mar,
	year = {2020},
	pages = {103342},
}

@article{modica_gradient_1987,
	title = {The gradient theory of phase transitions and the minimal interface criterion},
	volume = {98},
	issn = {0003-9527, 1432-0673},
	url = {http://link.springer.com/10.1007/BF00251230},
	doi = {10.1007/BF00251230},
	language = {en},
	number = {2},
	urldate = {2023-10-08},
	journal = {Archive for Rational Mechanics and Analysis},
	author = {Modica, Luciano},
	month = jun,
	year = {1987},
	pages = {123--142},
}

@misc{mottram_introduction_2014,
	title = {Introduction to {Q}-tensor theory},
	url = {http://arxiv.org/abs/1409.3542},
	language = {en},
	urldate = {2023-10-09},
	publisher = {arXiv},
	author = {Mottram, Nigel J. and Newton, Christopher J. P.},
	month = sep,
	year = {2014},
	note = {arXiv:1409.3542 [cond-mat]},
}

@article{spencer_zenithal_2010,
	title = {Zenithal bistable device: {Comparison} of modeling and experiment},
	volume = {82},
	issn = {1539-3755, 1550-2376},
	shorttitle = {Zenithal bistable device},
	url = {https://link.aps.org/doi/10.1103/PhysRevE.82.021702},
	doi = {10.1103/PhysRevE.82.021702},
	language = {en},
	number = {2},
	urldate = {2023-10-25},
	journal = {Physical Review E},
	author = {Spencer, T. J. and Care, C. M. and Amos, R. M. and Jones, J. C.},
	month = aug,
	year = {2010},
	pages = {021702},
}

@article{kitson_controllable_2002,
	title = {Controllable alignment of nematic liquid crystals around microscopic posts: {Stabilization} of multiple states},
	volume = {80},
	issn = {0003-6951, 1077-3118},
	shorttitle = {Controllable alignment of nematic liquid crystals around microscopic posts},
	url = {https://pubs.aip.org/apl/article/80/19/3635/514957/Controllable-alignment-of-nematic-liquid-crystals},
	doi = {10.1063/1.1478778},
	language = {en},
	number = {19},
	urldate = {2023-10-25},
	journal = {Applied Physics Letters},
	author = {Kitson, Stephen and Geisow, Adrian},
	month = may,
	year = {2002},
	pages = {3635--3637},
}

@book{boyd_convex_2004,
	address = {Cambridge, UK ; New York},
	title = {Convex {Optimization}},
	isbn = {978-0-521-83378-3},
	language = {en},
	publisher = {Cambridge University Press},
	author = {Boyd, Stephen P. and Vandenberghe, Lieven},
	year = {2004},
}

@article{golovaty_dimension_2015,
	title = {Dimension reduction for the {Landau}-de {Gennes} model in planar nematic thin films},
	volume = {25},
	issn = {0938-8974, 1432-1467},
	url = {http://link.springer.com/10.1007/s00332-015-9264-7},
	doi = {10.1007/s00332-015-9264-7},
	language = {en},
	number = {6},
	urldate = {2024-02-21},
	journal = {Journal of Nonlinear Science},
	author = {Golovaty, Dmitry and Montero, José Alberto and Sternberg, Peter},
	month = dec,
	year = {2015},
	pages = {1431--1451},
}

@article{wang_liquid_2018,
	title = {Liquid crystalline tactoids: ordered structure, defective coalescence and evolution in confined geometries},
	volume = {376},
	issn = {1364-503X, 1471-2962},
	shorttitle = {Liquid crystalline tactoids},
	url = {https://royalsocietypublishing.org/doi/10.1098/rsta.2017.0042},
	doi = {10.1098/rsta.2017.0042},
	language = {en},
	number = {2112},
	urldate = {2024-02-23},
	journal = {Philosophical Transactions of the Royal Society A: Mathematical, Physical and Engineering Sciences},
	author = {Wang, Pei-Xi and MacLachlan, Mark J.},
	month = feb,
	year = {2018},
	pages = {20170042},
}

@article{majumdar_equilibrium_2010,
	title = {Equilibrium order parameters of nematic liquid crystals in the {Landau}-de {Gennes} theory},
	volume = {21},
	issn = {0956-7925, 1469-4425},
	url = {https://www.cambridge.org/core/product/identifier/S0956792509990210/type/journal_article},
	doi = {10.1017/S0956792509990210},
	language = {en},
	number = {2},
	urldate = {2024-03-07},
	journal = {European Journal of Applied Mathematics},
	author = {Majumdar, Apala},
	month = apr,
	year = {2010},
	pages = {181--203},
}

@article{van_der_waals_thermodynamic_1979,
	title = {The thermodynamic theory of capillarity under the hypothesis of a continuous variation of density},
	volume = {20},
	issn = {0022-4715, 1572-9613},
	url = {http://link.springer.com/10.1007/BF01011514},
	doi = {10.1007/BF01011514},
	language = {en},
	number = {2},
	urldate = {2024-03-31},
	journal = {Journal of Statistical Physics},
	author = {van der Waals, J. D.},
	translator = {Rowlinson, J. S.},
	month = feb,
	year = {1979},
	pages = {200--244},
}

@article{sternberg_effect_1988,
	title = {The effect of a singular perturbation on nonconvex variational problems},
	volume = {101},
	issn = {0003-9527, 1432-0673},
	url = {http://link.springer.com/10.1007/BF00253122},
	doi = {10.1007/BF00253122},
	language = {en},
	number = {3},
	urldate = {2024-03-31},
	journal = {Archive for Rational Mechanics and Analysis},
	author = {Sternberg, Peter},
	month = sep,
	year = {1988},
	pages = {209--260},
}

@article{cahn_free_1958,
	title = {Free energy of a nonuniform system. {I}. {Interfacial} free energy},
	volume = {28},
	issn = {0021-9606, 1089-7690},
	url = {https://pubs.aip.org/jcp/article/28/2/258/74794/Free-Energy-of-a-Nonuniform-System-I-Interfacial},
	doi = {10.1063/1.1744102},
	language = {en},
	number = {2},
	urldate = {2024-04-18},
	journal = {The Journal of Chemical Physics},
	author = {Cahn, John W. and Hilliard, John E.},
	month = feb,
	year = {1958},
	pages = {258--267},
}

@article{onsager_effects_1949,
	title = {The effects of shape on the interaction of colloidal particles},
	volume = {51},
	issn = {0077-8923, 1749-6632},
	url = {https://nyaspubs.onlinelibrary.wiley.com/doi/10.1111/j.1749-6632.1949.tb27296.x},
	doi = {10.1111/j.1749-6632.1949.tb27296.x},
	language = {en},
	number = {4},
	urldate = {2024-07-10},
	journal = {Annals of the New York Academy of Sciences},
	author = {Onsager, Lars},
	month = may,
	year = {1949},
	pages = {627--659},
}

@article{de_gennes_short_1971,
	title = {Short range order effects in the isotropic phase of nematics and cholesterics},
	volume = {12},
	issn = {0026-8941},
	url = {https://www.tandfonline.com/doi/full/10.1080/15421407108082773},
	doi = {10.1080/15421407108082773},
	language = {en},
	number = {3},
	urldate = {2024-07-11},
	journal = {Molecular Crystals and Liquid Crystals},
	author = {de Gennes, P. G.},
	month = feb,
	year = {1971},
	pages = {193--214},
}

@article{liu_axial_2005,
	title = {Axial symmetry and classification of stationary solutions of {Doi}-{Onsager} equation on the sphere with {Maier}-{Saupe} potential},
	volume = {3},
	issn = {15396746, 19450796},
	url = {https://link.intlpress.com/JDetail/1806266493467840514},
	doi = {10.4310/CMS.2005.v3.n2.a7},
	language = {en},
	number = {2},
	urldate = {2025-03-13},
	journal = {Communications in Mathematical Sciences},
	author = {Liu, Hailiang and Zhang, Hui and Zhang, Pingwen},
	year = {2005},
	pages = {201--218},
}

@article{maier_einfache_1959,
	title = {Eine einfache molekular-statistische {Theorie} der nematischen kristallinflüssigen {Phase}. {Teil} {I}},
	volume = {14},
	copyright = {http://creativecommons.org/licenses/by-nc-nd/3.0/},
	issn = {1865-7109, 0932-0784},
	url = {https://www.degruyter.com/document/doi/10.1515/zna-1959-1005/html},
	doi = {10.1515/zna-1959-1005},
	language = {en},
	number = {10},
	urldate = {2025-03-31},
	journal = {Zeitschrift für Naturforschung A},
	author = {Maier, Wilhelm and Saupe, Alfred},
	month = oct,
	year = {1959},
	pages = {882--889},
}

@article{ball_nematic_2010,
	title = {Nematic liquid crystals: from {Maier}-{Saupe} to a continuum theory},
	volume = {525},
	issn = {1542-1406, 1563-5287},
	shorttitle = {Nematic liquid crystals},
	url = {http://www.tandfonline.com/doi/abs/10.1080/15421401003795555},
	doi = {10.1080/15421401003795555},
	language = {en},
	number = {1},
	urldate = {2025-04-15},
	journal = {Molecular Crystals and Liquid Crystals},
	author = {Ball, John M. and Majumdar, Apala},
	month = jul,
	year = {2010},
	pages = {1--11},
}

@article{feng_closure_1998,
	title = {Closure approximations for the {Doi} theory: {Which} to use in simulating complex flows of liquid-crystalline polymers?},
	volume = {42},
	issn = {0148-6055, 1520-8516},
	shorttitle = {Closure approximations for the {Doi} theory},
	url = {https://pubs.aip.org/sor/jor/article/42/5/1095-1119/239264},
	doi = {10.1122/1.550920},
	language = {en},
	number = {5},
	urldate = {2025-04-15},
	journal = {Journal of Rheology},
	author = {Feng, J. and Chaubal, C. V. and Leal, L. G.},
	month = sep,
	year = {1998},
	pages = {1095--1119},
}

@article{li_local_2015,
	title = {Local well-posedness and small {Deborah} limit of a molecule-based \textit{{Q}}-tensor system},
	volume = {20},
	issn = {1531-3492},
	url = {http://www.aimsciences.org/journals/displayArticlesnew.jsp?paperID=11548},
	doi = {10.3934/dcdsb.2015.20.2611},
	language = {en},
	number = {8},
	urldate = {2025-04-16},
	journal = {Discrete and Continuous Dynamical Systems - Series B},
	author = {Li, Sirui and Wang, Wei and Zhang, Pingwen},
	month = aug,
	year = {2015},
	pages = {2611--2655},
}

@article{bingham_antipodally_1974,
	title = {An antipodally symmetric distribution on the sphere},
	volume = {2},
	issn = {0090-5364},
	url = {https://projecteuclid.org/journals/annals-of-statistics/volume-2/issue-6/An-Antipodally-Symmetric-Distribution-on-the-Sphere/10.1214/aos/1176342874.full},
	doi = {10.1214/aos/1176342874},
	number = {6},
	urldate = {2025-04-22},
	journal = {The Annals of Statistics},
	author = {Bingham, Christopher},
	month = nov,
	year = {1974},
	pages = {1201--1225},
}

@article{kent_asymptotic_1987,
	title = {Asymptotic expansions for the {Bingham} distribution},
	volume = {36},
	issn = {00359254},
	url = {https://www.jstor.org/stable/10.2307/2347545?origin=crossref},
	doi = {10.2307/2347545},
	language = {en},
	number = {2},
	urldate = {2025-04-22},
	journal = {Applied Statistics},
	author = {Kent, John T.},
	year = {1987},
	pages = {139},
}

@article{chaubal_closure_1998,
	title = {A closure approximation for liquid-crystalline polymer models based on parametric density estimation},
	volume = {42},
	issn = {0148-6055, 1520-8516},
	url = {https://pubs.aip.org/sor/jor/article/42/1/177-201/239183},
	doi = {10.1122/1.550887},
	language = {en},
	number = {1},
	urldate = {2025-04-25},
	journal = {Journal of Rheology},
	author = {Chaubal, Charu V. and Leal, L. Gary},
	month = jan,
	year = {1998},
	pages = {177--201},
}

@article{weady_fast_2022,
	title = {A fast {Chebyshev} method for the {Bingham} closure with application to active nematic suspensions},
	volume = {457},
	issn = {00219991},
	url = {https://linkinghub.elsevier.com/retrieve/pii/S0021999121008329},
	doi = {10.1016/j.jcp.2021.110937},
	language = {en},
	urldate = {2025-04-28},
	journal = {Journal of Computational Physics},
	author = {Weady, Scott and Shelley, Michael J. and Stein, David B.},
	month = may,
	year = {2022},
	pages = {110937},
}

@article{chen_maximum_2021,
	title = {Maximum likelihood estimation of the {Fisher}–{Bingham} distribution via efficient calculation of its normalizing constant},
	volume = {31},
	issn = {0960-3174, 1573-1375},
	url = {https://link.springer.com/10.1007/s11222-021-10015-9},
	doi = {10.1007/s11222-021-10015-9},
	language = {en},
	number = {4},
	urldate = {2025-05-23},
	journal = {Statistics and Computing},
	author = {Chen, Yici and Tanaka, Ken'ichiro},
	month = jul,
	year = {2021},
	pages = {40},
}

@article{katriel_free_1986,
	title = {Free energies in the {Landau} and molecular field approaches},
	volume = {1},
	issn = {0267-8292, 1366-5855},
	url = {http://www.tandfonline.com/doi/abs/10.1080/02678298608086667},
	doi = {10.1080/02678298608086667},
	language = {en},
	number = {4},
	urldate = {2025-06-08},
	journal = {Liquid Crystals},
	author = {Katriel, J. and Kventsel, G. F. and Luckhurst, G. R. and Sluckin, T. J.},
	month = jul,
	year = {1986},
	pages = {337--355},
}

@article{lewis_convex_1996,
	title = {Convex analysis on the {Hermitian} matrices},
	volume = {6},
	issn = {1052-6234, 1095-7189},
	url = {http://epubs.siam.org/doi/10.1137/0806009},
	doi = {10.1137/0806009},
	language = {en},
	number = {1},
	urldate = {2025-06-09},
	journal = {SIAM Journal on Optimization},
	author = {Lewis, A. S.},
	month = feb,
	year = {1996},
	pages = {164--177},
}

@book{abramowitz_handbook_2013,
	address = {New York, NY},
	edition = {9th},
	series = {Dover books on mathematics},
	title = {Handbook of {Mathematical} {Functions} with {Formulas}, {Graphs}, and {Mathematical} {Tables}},
	isbn = {978-0-486-61272-0},
	shorttitle = {Handbook of {Mathematical} {Functions}},
	language = {eng},
	publisher = {Dover Publications},
	editor = {Abramowitz, Milton and Stegun, Irene A.},
	year = {2013},
}

@article{mardia_statistics_1975,
	title = {Statistics of directional data},
	volume = {37},
	copyright = {https://academic.oup.com/journals/pages/open\_access/funder\_policies/chorus/standard\_publication\_model},
	issn = {1369-7412, 1467-9868},
	url = {https://academic.oup.com/jrsssb/article/37/3/349/7027430},
	doi = {10.1111/j.2517-6161.1975.tb01550.x},
	language = {en},
	number = {3},
	urldate = {2025-06-12},
	journal = {Journal of the Royal Statistical Society Series B: Statistical Methodology},
	author = {Mardia, K. V.},
	month = jul,
	year = {1975},
	pages = {349--371},
}

@article{urbanski_liquid_2017,
	title = {Liquid crystals in micron-scale droplets, shells and fibers},
	volume = {29},
	issn = {0953-8984, 1361-648X},
	url = {https://iopscience.iop.org/article/10.1088/1361-648X/aa5706},
	doi = {10.1088/1361-648X/aa5706},
	language = {en},
	number = {13},
	urldate = {2025-06-26},
	journal = {Journal of Physics: Condensed Matter},
	author = {Urbanski, Martin and Reyes, Catherine G. and Noh, JungHyun and Sharma, Anshul and Geng, Yong and Subba Rao Jampani, Venkata and Lagerwall, Jan P. F.},
	month = apr,
	year = {2017},
	pages = {133003},
}

@article{wu_diffuse-interface_2025,
	title = {A diffuse-interface {Landau}–de {Gennes} model for free boundary problems in the theory of nematic liquid crystals},
	volume = {57},
	issn = {0036-1410, 1095-7154},
	url = {https://epubs.siam.org/doi/10.1137/24M1710322},
	doi = {10.1137/24m1710322},
	language = {en},
	number = {4},
	urldate = {2025-08-04},
	journal = {SIAM Journal on Mathematical Analysis},
	publisher = {Society for Industrial \& Applied Mathematics (SIAM)},
	author = {Wu, Dawei and Shi, Baoming and Han, Yucen and Zhang, Pingwen and Majumdar, Apala and Zhang, Lei},
	month = aug,
	year = {2025},
	pages = {4358--4395},
}

@article{fatkullin_critical_2005,
	title = {Critical points of the {Onsager} functional on a sphere},
	volume = {18},
	issn = {0951-7715, 1361-6544},
	url = {https://iopscience.iop.org/article/10.1088/0951-7715/18/6/008},
	doi = {10.1088/0951-7715/18/6/008},
	language = {en},
	number = {6},
	urldate = {2025-08-20},
	journal = {Nonlinearity},
	publisher = {IOP Publishing},
	author = {Fatkullin, I. and Slastikov, V.},
	month = nov,
	year = {2005},
	pages = {2565--2580},
}

@article{weady_thermodynamically_2022,
	title = {Thermodynamically consistent coarse-graining of polar active fluids},
	volume = {7},
	copyright = {https://link.aps.org/licenses/aps-default-license},
	issn = {2469-990X},
	url = {https://link.aps.org/doi/10.1103/PhysRevFluids.7.063301},
	doi = {10.1103/physrevfluids.7.063301},
	language = {en},
	number = {6},
	urldate = {2025-08-28},
	journal = {Physical Review Fluids},
	publisher = {American Physical Society (APS)},
	author = {Weady, Scott and Stein, David B. and Shelley, Michael J.},
	month = jun,
	year = {2022},
}

@book{priestley_introduction_1975,
	address = {Boston, MA},
	title = {Introduction to {Liquid} {Crystals}},
	copyright = {https://www.springernature.com/gp/researchers/text-and-data-mining},
	isbn = {978-1-4684-2177-4 978-1-4684-2175-0},
	url = {https://link.springer.com/10.1007/978-1-4684-2175-0},
	doi = {10.1007/978-1-4684-2175-0},
	language = {en},
	urldate = {2025-08-29},
	publisher = {Springer US},
	editor = {Priestley, E. B. and Wojtowicz, Peter J. and Sheng, Ping},
	year = {1975},
}

@article{bawden_liquid_1936,
	title = {Liquid crystalline substances from virus-infected plants},
	volume = {138},
	copyright = {http://www.springer.com/tdm},
	issn = {0028-0836, 1476-4687},
	url = {https://www.nature.com/articles/1381051a0},
	doi = {10.1038/1381051a0},
	language = {en},
	number = {3503},
	urldate = {2026-01-13},
	journal = {Nature},
	author = {Bawden, F. C. and Pirie, N. W. and Bernal, J. D. and Fankuchen, I.},
	month = dec,
	year = {1936},
	pages = {1051--1052},
}

@article{shi_neural_2026,
	title = {Neural network-based tensor models for liquid crystals with molecular-level information},
	volume = {113},
	issn = {2470-0045, 2470-0053},
	url = {https://link.aps.org/doi/10.1103/7v32-lr9w},
	doi = {10.1103/7v32-lr9w},
	language = {en},
	number = {1},
	urldate = {2026-01-18},
	journal = {Physical Review E},
	author = {Shi, Baoming and Majumdar, Apala and Zhang, Lei},
	month = jan,
	year = {2026},
	pages = {015401},
}

@article{wu_unlocking_2026,
	title = {Unlocking hidden topological multistability via biphasic correlated order evolution},
	volume = {136},
	issn = {0031-9007, 1079-7114},
	url = {https://link.aps.org/doi/10.1103/zyy7-cm33},
	doi = {10.1103/zyy7-cm33},
	language = {en},
	number = {6},
	urldate = {2026-02-14},
	journal = {Physical Review Letters},
	author = {Wu, Jin-Bing and Guo, Zhenghao and Shi, Baoming and Luo, Daoxing and Zhang, Lei and Lu, Yan-Qing and Hu, Wei},
	month = feb,
	year = {2026},
	pages = {068101},
}

@unpublished{wu_analysing_2026,
  title   = {Analysing the nematic liquid crystal droplet with an improved diffuse-interface {Landau}-de {Gennes} model},
  note = {to appear in SIAM Journal on Applied Mathematics},
  author  = {Shi, Baoming and Wu, Dawei and Zhang, Pingwen and Zhang, Lei}
}

@article{hornik_approximation_1991,
	title = {Approximation capabilities of multilayer feedforward networks},
	volume = {4},
	copyright = {https://www.elsevier.com/tdm/userlicense/1.0/},
	issn = {08936080},
	url = {https://linkinghub.elsevier.com/retrieve/pii/089360809190009T},
	doi = {10.1016/0893-6080(91)90009-T},
	language = {en},
	number = {2},
	urldate = {2025-05-29},
	journal = {Neural Networks},
	author = {Hornik, Kurt},
	year = {1991},
	pages = {251--257},
}

\appendix

\section{Properties of the quasi-entropy} \label{app:qe}

The quasi-entropy $\hat S(q)$ captures the blow-up behaviour of the entropy in limiting regimes, so it is expected to retain some key properties of the original entropy, e.g.~the isotropic–nematic phase transition. 
In this case, the explicit form of the quasi-entropy can allow quantitative analysis.

We replace the entropy $S(q)$ under 3D Bingham closure with the quasi-entropy $\hat S(q)$ \eqref{Shat-q} in the bulk energy \eqref{bing-bulk}:
\begin{equation} \label{qe-bulk}
    \hat F_b(q)=\hat S(q)- \frac{\alpha}{2}|q|^2,\ q_1+q_2+q_3=0.
\end{equation}
(written in terms of the eigenvalues $q=\lambda(\vb Q)$)
Then, we show that the new bulk energy still possesses qualitatively similar properties to the original energy \eqref{bing-bulk}, in the sense that as the temperature-related constant $\alpha$ changes, stationary points of the $\hat F_b$ exhibit a first-order isotropic-nematic phase transition. 

A stationary point $q=[q_1,q_2,q_3]^T$ of \eqref{qe-bulk} solves the equations
\begin{equation}
    -\frac12 \frac{1}{q_i+\frac13} - \alpha q_i = \lambda,
\end{equation}
where $\lambda$ is a Lagrangian multiplier from the constraint $q_1+q_2+q_3=0$. We then find that
\begin{equation} \label{qe-lag}
    \alpha q_i^2 +\qty(\frac{\alpha}{3} +\lambda) q_i + \frac12+\frac{\lambda}{3} = 0,
\end{equation}
so $q_i$ are roots of a polynomial with degree no more than 2, and take at most two values. In other words, all stationary points are uniaxial, like \eqref{bing-bulk} and the LdG energy.


In order to study the stationary points of $\hat F_b$ and their stability more accurately, we parametrize $q\in \Delta_0$ with
\begin{equation}
    q=\qty[ \frac{2s}{3}, -\frac{s+\eta}{3}, -\frac{s-\eta}{3}]^T,
\end{equation}
and then
\begin{equation}
    \hat F_b = -\frac12\left[ \log\frac{2s+1}{3} + \log \frac{1-s-\eta}{3}+ \log \frac{1-s+\eta}{3} \right] - \frac{\alpha}{2}\left( \frac23 s^2 + \frac29 \eta^2 \right).
\end{equation}
By the symmetry of eigenvalues, we can assume w.l.o.g. that $|q_1|\ge |q_2| \ge |q_3|$, which gives $|s|\ge |\eta|\ge 0$. At the stationary point, it must hold that $\eta=0$ because $q_2=q_3$. Equation \eqref{qe-lag} then becomes
\[\begin{cases}
    -\frac{1}{2s+1} +\frac{1}{2(1-s-\eta)}+\frac{1}{2(1-s+\eta)} - \frac23\alpha s=0 \\
    \frac{1}{2(1-s-\eta)}-\frac{1}{2(1-s+\eta)} - \frac29\alpha \eta=0
\end{cases}\]
The second equation follows from $\eta=0$ directly, and the first simplifies into
\[-\frac{1}{2s+1} + \frac{1}{1-s} - \frac23 \alpha s =0
\Rightarrow s (4\alpha s^2 -2\alpha s + 9-2\alpha)=0,\]
which gives three solutions:
\begin{equation} \label{qe-s}
    s_0=0,\ s_\pm = \frac{\alpha \pm 3\sqrt{\alpha(\alpha-4)}}{4\alpha}.
\end{equation}
We denote the stationary points corresponding to $s_{\{0,\pm\}}$ by $q_{\{0,\pm\}}$. $q_0=0$ is the isotropic phase, and $q_+$ is the uniaxial state corresponding to the nematic phase. 
Note that the physical constraint $q_\pm\in\Delta_0$ requires $s_\pm\in (-\frac12,1)$, so $q_\pm$ are well-defined only if $\alpha>4$, i.e.~nontrivial stationary points exist only at lower temperatures.

To study their stability, we compute the Hessian $H(q)=\nabla^2 \hat F_b(s,\eta)$ at $s=s_{\{0,\pm\}}$ and $\eta=0$:
\begin{equation}
    H(q)=\begin{bmatrix}
        \frac{2}{(2s+1)^2} + \frac{1}{(1-s)^2} - \frac23\alpha &
        0 \\ 0 & \frac{1}{(1-s)^2}-\frac29\alpha
    \end{bmatrix}.
\end{equation}
The stability of $q_{\{0,\pm\}}$ is determined by the signs of the diagonal entries of $H$.
At $q_0$, we get
\begin{equation}
    H_{11}(q_0)=3-\frac23\alpha,\ H_{22}(q_0) = 1-\frac29\alpha,
\end{equation}
so it is stable when $\alpha<\frac92$, and unstable when $\alpha>\frac92$. In other words, the isotropic state loses stability at a lower temperature.
At $q_{\pm}$,
we use \eqref{qe-s} to compute
\begin{equation}
\begin{aligned}
    H_{11}(q_\pm) &=\frac19 \alpha \qty (3\alpha-12\pm \sqrt{\alpha(\alpha-4)}), \\
    H_{22}(q_\pm) &= \frac29 \alpha\qty( \alpha-3 \pm \sqrt{\alpha(\alpha-4)} ).
\end{aligned}
\end{equation}
When $\alpha> 4$, we have that
\begin{gather*}
    H_{11}(q_+)>0,\ H_{22}(q_+)>0,\\
H_{11}(q_-)\begin{cases}
    <0, & \alpha<\frac92, \\
    >0, & \alpha>\frac92,
\end{cases}\ 
H_{22}(q_-)\begin{cases}
    <0, & \alpha>\frac92,\\
    >0, &\alpha<\frac92,
\end{cases}
\end{gather*}
so $q_+$ is always stable, while $q_-$ is always unstable (a saddle point of order 1).

When $4<\alpha<\frac92$, $q_0$ and $q_+$ coexist as stable local minimizers. Compare their energy levels:
\begin{equation}
\hat F_b(q_+)-\hat F_b(q_0) = - \frac12 \ln\left( \frac{9(1-s_+)}{2\alpha} \right)-\frac{\alpha s_+^2}{3}. \
\end{equation}
We assert that $\hat F_b(q_+)-\hat F_b(q_0)$ is strictly decreasing with respect to $\alpha$.
Let $\alpha=(t+\frac1t)^2$ with $t\ge 1$ (so $\alpha=\alpha(t)$ is monotonic), and then
\begin{align*}
    s_+&= \frac14 + \frac34 \sqrt{\frac{\alpha-4}{\alpha}} = \frac{2t^2-1}{2(t^2+1)}, \\
    \hat F_b(q_+)-\hat F_b(q_0)
&=-\frac12 \ln\qty(\frac{27}{4}\frac{1}{t^2+1}\cdot \frac{1}{(t+\frac1t)^2}) - \frac{1}{12} \qty(t+\frac1t)^2 \left( \frac{2t^2-1}{t^2+1} \right)^2 \\
&= -\frac12\ln \frac{27}{4} + \frac32\ln(t^2+1)-\ln t - \frac{1}{12} \qty(2t-\frac1t)^2 \triangleq \vph(t).
\end{align*}
Then, we compute the derivative of $\vph(t)$:
\[\vph'(t)=\frac{3t}{t^2+1}-\frac1t - \frac16\qty(2t-\frac1t)\qty(2+\frac{1}{t^2})=-\frac{(2t^2-1)^2(t^2-1)}{6t^3(t^2+1)}<0,\ \forall t>1,\]
so $\vph(t)$ is strictly decreasing with respect to $t$, and also to $\alpha$. In addition, $\vph(1)=\frac12\ln\frac{32}{27}-\frac{1}{12}>0$, and $\vph(t)\to-\infty$ as $t\to+\infty$, so there exists a unique zero $t^*>1$.
Numerical computation shows that the corresponding $\alpha^*=\alpha(t^*)\approx 4.051407.$

In summary, we have the following characterization of the stability results of the energy \eqref{qe-bulk}. The results are also shown in Table \ref{tab:Fb-sur-stab}.
\begin{proposition} \label{prop:qe-bulk-stab}
All stationary points of $\hat F_b(q)$ are uniaxial, i.e.~having two identical eigenvalues. Their existence and stability are given as follows:
\begin{enumerate}[\rm (i)]
\item When $\alpha<4$, the isotropic solution $q_0=0$ is the only stationary point of $\hat F_b$; when $\alpha\ge 4$, there also exists nonzero stationary points $q_\pm=s_\pm[\frac32,-\frac13,-\frac13]^T$, where $s_\pm$ are given by \eqref{qe-s}.
\item The solution $q_-$ is always unstable.
\item When $4<\alpha<\frac92$, $q_+$ and $q_0$ are both stable. When $\alpha>\frac92$, $q_+$ is the only stable stationary point, which corresponds to the nematic phase.
\item A first-order phase transition occurs at $\alpha^*\approx 4.051407,$ such that $q_+$ is the global minimizer when $\alpha>\alpha^*$, and vice versa when $\alpha<\alpha^*$.
\end{enumerate}
\end{proposition}

\begin{table}[th] 
\centering
\caption{Stability of stationary points of $\hat F_b$} \label{tab:Fb-sur-stab}
\begin{tabular}{cccc}
    \toprule
    \multirow{2}{*}{$\alpha$}{} & \multicolumn{3}{c}{Stationary points} \\
    \cline{2-4}
    & $q_0$ & $q_-$ & $q_+$ \\
    \midrule
    $(0,4)$ & global minimizer & N/A & N/A \\
    $(4,\alpha^*)$ & global minimizer & unstable & local minimizer \\
    $(\alpha^*,\frac92)$ & local minimizer & unstable & global minimizer \\
    $(\frac92, \infty)$ & unstable & unstable & global minimizer \\
    \bottomrule
\end{tabular}
\end{table}

\end{document}